\documentclass[
	a4paper,
	UKenglish,
	cleveref,
	autoref,
	numberwithinsect,
	final
]{lipics-compressed}

\usepackage{xparse}
\usepackage{xpatch}
\usepackage{etoolbox}
\usepackage{environ}

\newcommand{\nocontentsline}[3]{}
\newcommand{\tocless}[2]{\bgroup\let\addcontentsline=\nocontentsline#1{#2}\egroup}

\makeatletter
\newcommand{\IfLabelExistsTF}[3]{\@ifundefined{r@#1}{#3}{#2}}
\makeatother

\usepackage[utf8]{inputenc}

\usepackage{babel}
\usepackage[all]{foreign}
\defnotforeign[wrt]{{\xperiodafter{w.r.t}}}

\usepackage{amsmath}
\usepackage{amsthm}
\usepackage{amsfonts}
\usepackage{amssymb}
\usepackage{latexsym}
\usepackage{stmaryrd}
\usepackage{wasysym}
\usepackage{centernot}
\usepackage{mathtools}
\usepackage{amscd}
\usepackage{proof-dashed}

\allowdisplaybreaks

\usepackage{thmtools}

\makeatletter
\renewcommand{\xRightarrow}[2][]{\ext@arrow 03{10}{10}\Rightarrowfill@{#1}{#2}}
\makeatother

\theoremstyle{claimstyle}
\newtheorem{fact}[claim]{Fact}

\if@cref@capitalise
	\Crefname{fact}{Fact}{Facts}
\else
	\crefname{fact}{fact}{facts}
	\Crefname{fact}{Fact}{Facts}
\fi

\usepackage{float}
\usepackage{wrapfig}
\usepackage[section]{placeins}

\usepackage{footnote}

\usepackage{tikz}
\usetikzlibrary{automata,3d,calc,arrows,positioning}

\tikzset{
	nat/.style={double,double equal sign distance},
	nat>/.style={nat,-implies},
	<nat/.style={nat,implies-},
	descr/.style={anchor=center,fill=white},
	crossing/.style={preaction={draw=white,-,line width=#1}},
	crossing/.default=2pt,
	extended/.style={shorten >=-#1, shorten <=-#1},
	extended/.default=2pt,
	text decoration line/.style={
		line width=.1ex,solid,
		rounded corners=0,
		round cap-round cap},
	capped line/.style={round cap-round cap},
	diagram/.style={
		font=\small,auto,scale=1.5,-latex,
		baseline=(current bounding box.center)}
}

\makeatletter
\newcommand{\customlabel}[4][0]{%
	\protected@write\@auxout{}{\lstring\newlabel{#3}{{#4}{\thepage}{#4}{#3}{}}}%
	\protected@write\@auxout{}{\lstring\newlabel{#3@cref}{{[#2][#1][#1]#4}{\thepage}}}%
}
\makeatother

\hypersetup{
	hidelinks,
	final,
}

\usepackage[
  sectionbib,
  square,
  numbers,
  sort&compress
  ]{natbib}
\newcommand{\cat}[1]{\textnormal{\textsf{#1}}\xspace}

\newcommand{\Set}{\cat{Set}}
\newcommand{\kl}{\mathcal{K}\mspace{-1mu}l}

\newcommand{\Sup}{\cat{Sup}}

\DeclareMathOperator*\supp{\rm{supp}}

\newcommand{\defeq}{\triangleq}

\newcommand{\rightdcirc}{\makebox[1.1\width][l]{\ensuremath{%
\longrightarrow%
\makebox{$\mkern-24mu\color{white}{\bullet}\mkern+12mu$}%
\makebox{$\mkern-21mu\circ\mkern+10mu$}%
\ignorespacesafterend}}}%

\title{Two modes of recognition: algebra, coalgebra, and languages}

\titlerunning{Two modes of recognition}

\author{Tomasz Brengos}
	{Faculty of Mathematics and Information Science, Warsaw University of Technology, ul. Koszykowa 75 00-662 Warszawa, Poland }
	{t.brengos@mini.pw.edu.pl}
	{} 
	{Supported by the grant of Warsaw University of Technology no.~504M for young researchers.}

\author{Marco Peressotti}
	{Department of Mathematics and Computer Science, University of Southern Denmark, Campusvej 55, DK-5230 Odense M, Denmark}
	{peressotti@imada.sdu.dk}
	{https://orcid.org/0000-0002-0243-0480}
  	{Partially supported by the Independent Research Fund Denmark, Natural Sciences, grant DFF-7014-00041.}

\authorrunning{T.~Brengos and M.~Peressotti}
\Copyright{Tomasz Brengos and Marco Peressotti}



\category{}

\relatedversion{}

\supplement{}

\nolinenumbers 

\hideLIPIcs  

\EventEditors{} \EventNoEds{0} \EventLongTitle{} \EventShortTitle{}
\EventAcronym{} \EventYear{2019} \EventDate{} \EventLocation{} \EventLogo{}
\SeriesVolume{} \ArticleNo{}

\begin{document}

\maketitle
\begin{abstract}\looseness=-1
The aim of the paper is to build a connection between two approaches towards categorical language theory: the coalgebraic and algebraic language theory for monads. For a pair of monads modelling the branching and the linear type we defined regular maps that generalize regular languages known in classical non-deterministic automata theory. These maps are behaviours of certain automata (\ie they possess a coalgebraic nature), yet they arise from Eilenberg-Moore algebras and their homomorphisms (by exploiting duality between the category of Eilenberg-Moore algebras and saturated coalgebras).
 
Given some additional assumptions, we show that regular maps form a certain subcategory of the Kleisli category for the monad which is the composition of the branching and linear type. Moreover, we state a Kleene-like theorem characterising the regular morphisms category in terms of the smallest subcategory closed under certain operations. Additionally, whenever the branching type monad is taken to be the powerset monad, we show that regular maps are described as maps recognized by certain functors whose codomains are categories with all finite hom-sets. 

We instantiate our framework on classical non-deterministic automata, tree automata, fuzzy automata and weighted automata.

\end{abstract}

\section{Introduction}
\label{section:introduction}
\looseness=-1
Automata theory is one of the core branches of theoretical computer science and formal language theory. One of the most fundamental state-based structures considered in the literature is a non-deterministic automaton and its relation with languages. Non-deterministic automata with a finite state-space are known to accept  \emph{regular languages}, characterized as subsets of words over a fixed finite alphabet that can be obtained from  the languages consisting of words of length less than or equal to one via a finite number of applications of three types of operations:  union, concatenation and the Kleene star operation \cite{Hopcroft:2000:IAT:557657}. This result is known under the name of \emph{Kleene theorem for regular languages}. It readily generalizes to automata accepting other types of input with more general versions of this theorem stated in the category-theoretic setting in the context of coalgebras and Lawvere theories \cite{esik2011,esik2013,esikhajgato2009:ai,brengos2018:concur,pin:automata}.
Coalgebraic language theory is based on a unifying theory of different types of automata and has been part of the focus of the coalgebraic community in recent years (\eg \cite{rutten:universal,jacobs08:cmcs,jacobssilvasokolova2012:cmcs,bonchi2015killing}). Our paper puts the main emphasis on a part of this research which describes a general theory of systems with internal transitions \cite{silva2013:calco,brengos2014:cmcs,brengos2015:lmcs,bonchi2015killing,brengos2016:concur,brengos2015:jlamp,mp2013:weak-arxiv}.
Intuitively, these systems have a special computation branch that is silent. This special branch, usually denoted by the letter $\tau$ or $\varepsilon$, is allowed to take several steps and in some sense remain neutral to the structure of a process. These systems arise in a natural manner in many branches of theoretical computer science, among which are process calculi \cite{milner:cc} (labelled transition systems with $\tau$-moves and their weak bisimulation) or automata theory (automata with $\varepsilon$-moves), to name only two. The approach from \cite{brengos2015:lmcs,brengos2015:jlamp} suggests that these systems should be defined as coalgebras whose type is a monad.  This treatment allows for an elegant modelling of weak behavioural equivalences \cite{brengos2015:jlamp,brengos2016:concur,brengos2019:lmcs} among which we find Milner's weak bisimulation  \cite{milner:cc}. Each coalgebra $\alpha\colon X\to TX$ becomes an endomorphism $\alpha\colon X\rightdcirc X$ in the Kleisli category for the monad $T$ and Milner's weak bisimulation on a labelled transition system $\alpha$ can be defined to be a strong bisimulation on its \emph{saturation} $\alpha^\ast$ which is the smallest LTS over the same state space satisfying $\alpha\leq \alpha^\ast$,  $\mathsf{id}\leq \alpha^\ast \text{ and }\alpha^\ast \cdot \alpha^\ast \leq \alpha^\ast$ (where the composition and the order are given in the Kleisli category for the LTS monad) \cite{brengos2015:lmcs}.
Hence, intuitively, $\alpha^\ast$ is the reflexive and transitive closure of $\alpha$.

Saturation $\alpha\mapsto \alpha^\ast$ can also be used as one of the main components of the coalgebraic language theory.  Indeed, the language accepted by an automaton whose transition map is modelled by $\alpha$ can be defined in terms of a simple expression involving its saturation $\alpha^\ast\colon X\to TX$ calculated in the Kleisli category for the monad $T$ (see \cite{bloomesik:93,esik2013,brengos2018:concur}). \emph{Regular} languages, \ie languages accepted by automata with finite carriers for carefully chosen transition $\alpha$ form a subclass of the class of all languages accepted by automata of type $T$.

\looseness=-1
Languages have also been studied from the algebraic perspective (\eg \cite{weil2004,Pin:1986:VFL:576708,eilenberg,Gradel:2002:ALI:938135,Wilke93,pin:automata}) with a general approach presented on the categorical level in the context of Eilenberg-Moore algebras for a monad in \cite{bojanczyk2015}.  For set-based algebras, a language (\ie a subset of the carrier of a given algebra) is said to be \emph{recognizable} if it is a preimage of a subset of a finite algebra under an algebra homomorphism. Using this approach one may \eg characterize regular languages for non-deterministic automata as recognizable languages for the monoid of words $(\Sigma^\ast,\cdot,\varepsilon)$. An algebraic characterization of classical regular languages is one of several examples of a similar phenomenon, where regular and recognizable languages meet (see \loccit).

\subparagraph{Contributions}\looseness=-1
We show existence of a general coincidence between an algebraic and coalgebraic approach towards defining languages stated on a categorical level by building on the duality between Eilenberg-Moore algebras and saturated coalgebras. In this setting, we define regular languages as a class of morphisms (herein, \emph{regular morphisms}) arising from automata whose coalgebra structure is saturated and is dual to an Eilenberg-Moore algebra. As we put our emphasis on automata with finite carriers, it is natural to consider Lawvere theories since a Lawvere theory for a monad is, roughly speaking, the part of its Kleisli category which is suitable to model morphisms with finite domains and codomains only \cite{lawvere:1963,hyland:power:2007}.  Lawvere theories become our natural habitat where  we provide Kleene-like theorem at the level of regular morphisms.
Additionally, in the case of generalized non-deterministic automata we show that regular languages (with variables) which are modelled by arrows in one Lawvere theory are essentially subsets of arrows of another Lawvere theory recognized by Lawvere theory morphisms whose targets are finitary theories. Hence, we obtain a general algebraic characterisation of such languages.

\section{Basic notions}\label{section:basic_notions}

We assume the reader is familiar with basic category theory concepts like a functor, a monad $(T,\mu,\eta)$, an adjunction, a Lawvere theory, a Kleisli category $\kl(T)$ and an Eilenberg-Moore category $\mathcal{EM}(T)$ for a monad $T$, a distributive law $\lambda\colon ST\implies TS$ of a monad $(S,m,e)$ over a monad $(T,\mu,\eta)$, a lifting of a monad $(S,m,e)$ to a monad $(\overline{S},\overline{m},\overline{e})$ on $\kl(T)$ and the fact that if $(S,m,e)$ lifts to $(\overline{S},\overline{m},\overline{e})$ on $\kl(T)$ then it yields a monadic structure on $TS$ whose Kleisli category satisfies $\kl(TS)=\kl(\overline{S})$ (see \eg \cite{maclane:cwm,barrwells:ttt,mulry:mfps1993} for details).

The most important example of a monad used throughout the paper is the \emph{powerset monad} $(\mathcal{P}\colon \Set\to \Set,\bigcup,\{-\})$. Moreover, we also consider the following running example.

\begin{example}\label{example:monoid_monad} \label{example:kleisli_powerset}\label{example:lts_monoid_lifting} \label{example:eilenberg_moore_algebras_for_monads}
\looseness=-1
Let $M=(M,\cdot,1)$ be any monoid. The functor $M\times \mathcal{I}d\colon \Set\to \Set$ carries a monadic $(M\times \mathcal{I}d,m,e)$ with $m_X\colon  M\times M\times X\to M\times X; (m,n,x)\mapsto (m\cdot n,x)$ and $e_X\colon X\to M\times X; x\mapsto (1,x)$.
The most often used example in our paper is the monad $\Sigma^\ast\times \mathcal{I}d$, where $\Sigma^\ast$ is the free monoid over a set $\Sigma$.
Eilenberg-Moore algebras for the monad $M\times \mathcal{I}d\colon\Set\to \Set$ consist of algebras $a\colon M\times X\to X$ satisfying $a(1,x) = x$ and $a(m\cdot n,x)=a(m,a(n,x))$ for any $x\in X$.
The monad $M\times \mathcal{I}d$ lifts to a monad $\overline{M}\colon\kl(\mathcal{P})\to \kl(\mathcal{P})$ via the distributive law $\theta\colon M\times \mathcal{P}\to \mathcal{P}(M\times \mathcal{I}d)$ given $\theta_X\colon M\times \mathcal{P}X\to \mathcal{P}(M\times X); (m,Y) \mapsto \{(m,y) \mid x \in Y\}$ \cite{brengos2015:lmcs,brengos2019:lmcs}. The monad $\overline{M}$ maps any object $X$ in $\kl(\mathcal{P})$ onto $\overline{M}X = M\times X$ and any map $f\colon X\to \mathcal{P}Y$ between $X$ and $Y$ in $\kl(\mathcal{P})$ onto $\overline{M}f = M\times X\stackrel{M\times f }{\to} M\times \mathcal{P}Y\stackrel{\theta_Y}{\to}\mathcal{P}(M\times Y)$. Its multiplication and unit are $\overline{m}=\{-\}\circ m$ and $\overline{e}=\{-\}\circ e$ respectively. The Kleisli category $\kl(\overline{M})$ has sets as objects and maps $X\to \mathcal{P}(M\times Y)$ as morphisms from $X$ to $Y$. The composition in the Kleisli category $\kl(\overline{M})$ is given for any $f\colon X\to \mathcal{P}(M\times Y)$ and $g\colon Y\to \mathcal{P}(M\times Z)$ as $(g\cdot f)(x) = \{(m_1\cdot m_2,z) \mid (m_2,z)\in g(y) \text{ and }(m_1,y)\in g(x) \}$.
Identity morphisms are the maps $x\mapsto \{(1,x)\}$. The lifting $\overline{M}$ of the monad $M\times \mathcal{I}d$ yields a monadic structure on the functor $\mathcal{P}(M\times \mathcal{I}d)$. For $M=\Sigma^\ast$, this monad $\mathcal{P}(\Sigma^\ast\times \mathcal{I}d)$ is called \emph{LTS monad} \cite{brengos2015:lmcs}.
Eilenberg-Moore algebras for the lifting $\overline{M}\colon \kl(\mathcal{P})\to \kl(\mathcal{P})$ of $M\times \mathcal{I}d$ to $\kl(\mathcal{P})$ are algebras whose underlying morphism is $a\colon \overline{M}X\rightdcirc X=M\times X\to \mathcal{P}X$,\footnote{In order to distinguish morphisms from the Kleisli category $\kl(T)$ and the base category $\mathsf{C}$ we often denote the former by $\rightdcirc$ and the latter by $\to$. Hence, $X\rightdcirc Y = X\to TY$ for any two objects.} where $a(1,x) = \{x\}$ and $a(m\cdot n,x) = \bigcup\{ a(n,y) \mid y \in  a(m,x)\}$.
\end{example}

\subparagraph{Lawvere theories}  The primary interest of the theory of automata and formal languages focuses on  automata over a \emph{finite} state space. Hence, since, as stated in the introduction, we are interested in systems with internal moves (\ie coalgebras $X\to TX$ for a monad $T$), without any loss of generality we may focus our attention on coalgebras of the form $n\to Tn$, where $n\defeq \{1,\ldots,n\}$ with $n=0,1,\ldots$ for a $\Set$-monad $T$. These morphisms are endomorphisms in a full subcategory of the Kleisli category for $T$ whose objects are $n$ for $n=0,1\ldots$ which is known under the name of \emph{(Lawvere) theory} and is denoted by $\mathbb{T}_T$. That is why we will often restrict the setting of this paper to Lawvere theories. Because we are interested in the coalgebraic essence of a Lawvere theory, we adopt the definition which is dual to the classical notion \cite{lawvere:1963}.

\subparagraph{Coalgebras and saturation}
\looseness=-1
Saturated coalgebras were introduced in \cite{brengos2014:cmcs,brengos2015:lmcs} in the context of coalgebraic weak bisimulation. As noticed in \emph{loc. cit.}  the concept of a saturated map can be given in any order enriched category\footnote{A category is \emph{order enriched} if each hom-set is a poset with the order being preserved by the composition.} $\mathsf{K}$: we say that an endomorphism $\alpha\colon X\to X$ in $\mathsf{K}$ is \emph{saturated} if $\mathsf{id}\leq \alpha\text{ and }\alpha\circ \alpha\leq \alpha$. Whenever $S=(S,m,e)$ is a monad then $\alpha\colon X\to SX$ is \emph{saturated} if the endomorphism $\alpha\colon X\rightdcirc X$ is saturated in the order enriched category $\kl(S)$. If we assume $(S,m,e)$ is a monad on an order enriched category $\mathsf{K}$ and $S$ is \emph{monotonic}\footnote{ $f\leq g \implies Sf\leq Sg$ for any pair of morphisms in $\mathsf{K}$ with a common domain and codomain.} then we can introduce an order on the category $\kl(S)$ which arises from the order enrichment of the base category $\mathsf{K}$ in an obvious way. In this case, the inequalities that define a saturated endomorphism can be translated into the language of  $\mathsf{K}$ by:
$e\leq \alpha \text{ and }m \circ S\alpha \circ \alpha \leq \alpha$. These two axioms bear resemblance to the axioms that define Eilenberg-Moore algebras for $S$. The purpose of Section~\ref{section:duality} is to elaborate more on this connection.

Let $\kl(S)$ be order enriched. By $\mathsf{Sat}(S)$ we denote the category whose objects are saturated $S$-coalgebras and morphisms are maps $f\colon X\to Y\in \mathsf{K}$ between the carriers of $\alpha\colon X\to SX$ and $\beta\colon Y\to SY$ which satisfy $Sf\circ \alpha \leq \beta \circ f$. Following \cite{brengos2015:lmcs,brengos2015:jlamp} we say that the monad $S$ \emph{admits saturation} if for any $S$-coalgebra $\alpha\colon X\to SX$ there is $\alpha^\ast\colon X\to SX\in \mathsf{Sat}(S)$ such that $\alpha^\ast$ is the smallest saturated coalgebra which satisfies $\alpha\leq \alpha^\ast$ and  $f\circ \alpha \Box \beta \circ Sf \implies f\circ \alpha^\ast \Box \beta^\ast \circ Sf$ for $\Box\in \{\leq ,\geq \}$ and any $f\colon X\to Y\in \mathsf{K}$.

\begin{example}\label{example:saturation}\label{example:saturation_for_LTS}
The monad $\overline{M}\colon \kl(\mathcal{P})\to \kl(\mathcal{P})$ from Ex. \ref{example:monoid_monad} admits saturation \cite{brengos2015:lmcs}. Given any $\alpha\colon X\rightdcirc \overline{M}X = X\to \mathcal{P}(M\times X)$ the saturated map $\alpha^\ast\colon X\to\mathcal{P}(M\times X)$ satisfies:
$x\stackrel{1}{\to}_{\alpha^\ast} x$  for any $x\in X$ and $x\stackrel{m_1\cdot \ldots \cdot m_k}{\longrightarrow}_{\alpha^\ast} x' \iff x\stackrel{m_1}{\to}_\alpha x_1\stackrel{m_1}{\to}_\alpha \ldots \stackrel{m_k}{\to}_\alpha x_k = x'$.
\end{example}

\section{Classical automata and regular languages, revisited}\label{section:classical_regular_revisited}

The main purpose of the section is to restate the basic properties and definitions from non-deterministic automata theory in the  language of category theory. We will elaborate more on the (co)algebraic characterisation of classical regular languages from this perspective. This section should serve as a more detailed introduction to the remaining part of the paper.

 A (finite non-deterministic) \emph{automaton} \cite{Hopcroft:2000:IAT:557657}  is a tuple $\mathcal{A}=(X,\delta \subseteq X\times \Sigma \times X,  \mathcal{F}\subseteq X)$, where $X$ is a finite set called the \emph{set of states}, $\delta$ is the \emph{transition} and $\mathcal{F}$ is the set of \emph{final states}.
 The language $L(\mathcal{A},x)$ of a state $x\in X$ in the automaton $\mathcal{A}$ is defined to be the set of words $\{w\in \Sigma^\ast \mid x\stackrel{w}{\to} x'\in \mathcal{F}\}$, where $x\stackrel{\varepsilon}{\to} x'$ iff $x=x'$ and  $x\stackrel{w}{\to}x' \stackrel{\Delta}{\iff} x\stackrel{a_1}{\to} x_1 \stackrel{a_2}{\to} x_2\ldots \stackrel{a_n}{\to} x_n=x'$ for $w=a_1\ldots a_n$ and $y\stackrel{a}{\to}y' \stackrel{\Delta}{\iff} (y,a,y')\in \delta$ for $y,y'\in X$ and $a\in \Sigma$ \footnote{Note that the textbook definition of an automaton usually includes the specification of the so-called \emph{initial state}  (see \eg \cite{Hopcroft:2000:IAT:557657}). Then the language of an automaton is defined to be the language of its initial state.}. Note that since $X$ is finite we can assume without any loss of generality that $X=n$ for some positive integer $n$. We can see that $\delta$ can be encoded by a map $\alpha\colon n\to \mathcal{P}(\Sigma\times n); i\mapsto \{(a,j)\mid i\stackrel{a}{\to}j \text{ in }\mathcal{A}\}$. Hence, the automaton $\mathcal{A}$ can be viewed as a pair $(\alpha\colon n\to \mathcal{P}(\Sigma\times n),\mathcal{F}\subseteq n).$

\subparagraph{Automata in categories} We will now focus on the categorical perspective on non-deterministic automata and their languages.  First, we introduce basic players of this paragraph and establish the notation. Here we work with two main categories, namely: $\kl(\mathcal{P})$ and $\Set$. 
These two categories share the class of objects: all sets. What is different is the morphisms and the compositions: we denote the morphisms from $\kl(\mathcal{P})$ by $\rightdcirc$ and the maps from $\Set$  by $\to$. 
\begin{wrapfigure}[4]{r}{0.2\textwidth}%
\vspace{-1.4ex}
\small$\infer{\infer{{\alpha\colon n\rightdcirc \overline{\Sigma^\ast} n}}{\alpha\colon n\to \mathcal{P}(\Sigma^\ast \times n)}}{\alpha\colon n\to \mathcal{P}(\Sigma\times n)}$%
\end{wrapfigure}
Hence, $X\rightdcirc Y = X\to \mathcal{P}Y$. Considering the fact that the monad   $(\Sigma^\ast\times \mathcal{I}d,m,e)$ lifts to the monad $(\overline{\Sigma^\ast},\overline{m},\overline{e})=(\overline{\Sigma^\ast},\{-\}\circ m,\{-\}\circ e)$ on $\kl(\mathcal{P})$ (as in Ex. \ref{example:monoid_monad}) the codomain of the transition map $\alpha$ of $\mathcal{A}$ changes depending on which category it is considered in---as summarised aside. Since the monad $\overline{\Sigma^\ast}$ admits saturation we also have the map $\alpha^\ast\colon n\rightdcirc \overline{\Sigma^\ast}n =n\to \mathcal{P}(\Sigma^\ast\times n)$ given by (\cf \cref{example:saturation}):
$
\alpha^\ast(i)= \{(\varepsilon,i)\}\cup \{(a_1\ldots a_k,j)\mid i\stackrel{a_1}{\to}_\alpha \ldots \stackrel{a_k}{\to}_\alpha i_k=j\}.
$
Note that there is an obvious bijection between the set of all languages $L\subseteq \Sigma^\ast$ and maps $1\to \mathcal{P}(\Sigma^\ast\times 1)$. This allows us to represent the language $L(\mathcal{A},i)$ of a state $i$ in the automaton $\mathcal{A}$ in terms of a morphism  $L(\alpha,\mathcal{F},i)\colon 1\to \mathcal{P}(\Sigma^\ast\times 1) = 1\rightdcirc \overline{\Sigma^\ast} 1$, which maps the unique element of $1$ onto $\{(w,1) \mid w\in L(\mathcal{A},i)\}$. It is easy to see that this language morphism can be expressed in terms of  a composition of maps calculated in $\kl(\mathcal{P})$ (\emph{conf.} Table~\ref{table:reg}).
\begin{table}
\resizebox{1\textwidth}{!}{
\begin{tabular}{|cc|c|}
& $L(\alpha,\mathcal{F},i)=$ &\emph{where} \\
 \hline
   &  & $\chi_\mathcal{F}\colon n\rightdcirc 1=n\to \mathcal{P}1,\  \chi_\mathcal{F}(i) = \left \{ \begin{array}{cc} \{1\} & \text{ if }i\in \mathcal{F}, \\ \varnothing  & \text{ otherwise} \end{array}\right.$\\
  & $1\stackrel{i_n}{\rightdcirc } n \stackrel{\alpha^\ast} {\rightdcirc } \overline{\Sigma^\ast}n \stackrel{\overline{\Sigma^\ast} \chi_{\mathcal{F}}}{\rightdcirc }\overline{\Sigma^\ast}1$ & $\overline{\Sigma^\ast}\chi_\mathcal{F}\colon \overline{\Sigma^\ast}n\rightdcirc \overline{\Sigma^\ast}1=\Sigma^\ast\times n\to \mathcal{P}(\Sigma^\ast\times 1),\  \overline{\Sigma^\ast}\chi_\mathcal{F}(w,i) = \left \{ \begin{array}{cc} \{(w,1)\} & \text{ if }i\in \mathcal{F}, \\ \varnothing  & \text{ otherwise}  \end{array}\right.$ \\ & & $i_n\colon 1\rightdcirc n=1\to \mathcal{P}n;1\mapsto \{i\}$.  \\ & & \\
  \hline
\end{tabular}
}
\\
\caption{ Languages of $(\alpha,\mathcal{F})$ expressed in  $\kl(\mathcal{P})$} \label{table:reg}
\vspace{-0.3cm}
\end{table}
\subparagraph{Algebra-coalgebra language coincidence} The entry in the first column of Table~\ref{table:reg} may be viewed as a coalgebraic (automata) definition of regular languages stated in the category $\kl(\mathcal{P})$.  Interestingly, it immediately allows us to see the dual, \emph{algebraic}, characterisation of these languages. Indeed, the category $\kl(\mathcal{P})$ comes with $(-)_-\colon \kl(\mathcal{P})\to \kl(\mathcal{P})^{op}$ mapping any object onto itself and any map $f\colon X\rightdcirc Y=X\to \mathcal{P}Y$ onto
\begin{align*}
f_-\colon Y\rightdcirc X = Y\to \mathcal{P}X; y\mapsto \{x\in X\mid y\in f(x)\}.\tag{OP}\label{eq:f_op}
\end{align*}
Additionally, it can be shown that  the functor $\overline{\Sigma^\ast}$ on $\kl(\mathcal{P})$ commutes with $(-)_-$, \ie $(\overline{\Sigma^\ast}f)_-= \overline{\Sigma^\ast}f_-$ for any $f\colon X\rightdcirc Y\in \kl(\mathcal{P})$. It turns out that $(\alpha^\ast)_-\colon \overline{\Sigma^\ast}n\rightdcirc  n$ in $\kl(\mathcal{P})$ is an Eilenberg-Moore algebra for the monad $(\overline{\Sigma^\ast},\overline{m},\overline{e})$ from Example \ref{example:kleisli_powerset}. Moreover, the $\kl(\mathcal{P})$-morphism $L(\alpha,\mathcal{F})=n \stackrel{\alpha^\ast} {\rightdcirc } \overline{\Sigma^\ast}n \stackrel{\overline{\Sigma^\ast} \chi_{\mathcal{F}}}{\rightdcirc }\overline{\Sigma^\ast}1=n\to \mathcal{P}(\Sigma^\ast\times 1)$
which maps $i$ to its language $L(\alpha,\mathcal{F},i)(1)=\{(w,1)\mid w\in L(\mathcal{A},i)\}$ satisfies the following statement.

\begin{fact}\label{fact:duality_regular}
The map $L(\alpha,\mathcal{F})_-\colon  \overline{\Sigma^\ast}1 \rightdcirc n=\Sigma^\ast\times 1\to \mathcal{P}n$, which maps a pair $(w,1)$ to the set of states of $(\alpha,\mathcal{F})$ that accept $w$, is an algebra homomorphism from the free Eilenberg-Moore algebra $\overline{m}_1\colon \overline{\Sigma^\ast}\overline{\Sigma^\ast}1\rightdcirc \overline{\Sigma^\ast} 1$ to the algebra $(\alpha^\ast)_-\colon \overline{\Sigma^\ast}n\rightdcirc n$. Additionally, any homomorphism from $\overline{m}_1$ to $(\alpha^\ast)_-$  in $\mathcal{EM}(\overline{\Sigma^\ast})$ is of the form $L(\alpha,\mathcal{F})_-$ for some $\mathcal{F}\subseteq n$.
\end{fact}

The above statement is a consequence of a more general Theorem \ref{lemma:preimage_homomorphism} stated in the next section. Since, as we will show in the remaining part of the paper, any Eilenberg-Moore algebra in $\mathcal{EM}(\overline{\Sigma^\ast})$ is of the form $(\alpha^\ast)_-$ for some morphism $\alpha\colon X\to \mathcal{P}(\Sigma\times X)$, the above fact can be read as follows: a language map $L\colon n\to \mathcal{P}(\Sigma^\ast\times 1)=n\rightdcirc \overline{\Sigma^\ast}1$ is regular (\ie $L=L(\alpha,\mathcal{F})$ for some automaton $(\alpha,\mathcal{F})$) if and only if its dual $L_-\colon \overline{\Sigma^\ast}1\rightdcirc n=\Sigma^\ast \times 1 \to \mathcal{P}n$ is an algebra homomorphism from $\overline{m_1}\colon \overline{\Sigma^\ast}\overline{\Sigma^\ast}1\rightdcirc \overline{\Sigma^\ast}1$ to an Eilenberg-Moore algebra over a finite carrier. Interestingly, this characterization leads us to the following result (see \cref{lemma:from_regular_to_morphisms} for a more general version).

\begin{fact}\label{fact:characterisation_regular_classically_morphisms} For $L\subseteq \mathbb{T}_{\Sigma^\ast\times \mathcal{I}d}(1,1)$ the map $\widehat{L}\colon 1\rightdcirc \overline{\Sigma^\ast}1=1\to \mathcal{P}(\Sigma^\ast\times 1); 1\mapsto \{l(1)\mid l\in L \}$ is regular  if and only if there is a Lawvere theory morphism $h\colon  \mathbb{T}_{\Sigma^\ast\times \mathcal{I}d}\to \mathbb{T}'$ into a finitary\footnote{A \emph{theory morphism} $h\colon \mathbb{T}\to \mathbb{T}'$ is a functor which maps $n$ onto itself. A theory $\mathbb{T}'$ is \emph{finitary} if all hom-sets $\mathbb{T}'(m,n)$ are finite.} Lawvere theory $\mathbb{T}'$ such that $L = h^{-1}(T')$ for some $T'\subseteq \mathbb{T}'(1,1)$.
\end{fact}

Since the restriction of $h$ to hom-sets: $\mathbb{T}_{\Sigma^\ast\times \mathcal{I}d}(1,1)$ and $\mathbb{T}'(1,1)$, is a monoid homomorphism from the monoid $(\mathbb{T}_{\Sigma^\ast\times \mathcal{I}d}(1,1),\circ, \mathsf{id})$ to the monoid $(\mathbb{T}'(1,1),\circ,\mathsf{id})$, the above statement may be viewed as a Lawvere theory generalization of the classical characterisation of regular languages as languages recognized by monoid homomorphisms.

The aim of the remaining part of the paper is to generalize these observations to arbitrary $\Set$-based monads (modulo some extra assumptions).

\subparagraph{Final remarks}
\label{remark:on_final_states}
Predominantly, in the coalgebraic literature finite behaviour (language) of systems is introduced in terms of the finite trace \cite{silva2013:calco,bonchi2015killing,jacobssilvasokolova2012:cmcs}. In the order enriched setting for which the type monad encodes terminal states, the finite trace is given by $\alpha^\dagger = \mu x.x\cdot \alpha$ \cite{brengos2014:cmcs}. However, in our setting the final states are not part of the transition and the language is defined via saturation. Although, as noted in \cite{brengos2018:concur,esikhajgato2009:ai} these two approaches are equivalent we choose our approach since it shows a more evident connection between the algebraic and coalgebraic frameworks for defining languages emphasizing the duality between Eilenberg-Moore algebras and (a subcategory of) saturated coalgebras. At this point the reader may also wonder why we choose Lawvere theories as the setting for our algebraic characterisation of languages (akin to Fact \ref{fact:characterisation_regular_classically_morphisms}). Indeed, such a treatment seems to be a redundant overcomplication in the light of a simple, monoid homomorphism characterisation. However, non-deterministic automata and regular languages in the classical sense revolve around sequential data.
If we move away from sequential data and deal with \eg trees then we need to be able to simultaneously consider  terms with more (but a finite number of) variables. We refer the reader to \eg \cite{brengos2018:concur} where a simple example to understand this phenomenon has been described in the context of  regular tree languages and an analogue of the Kleene theorem for trees.

\section{On algebra-coalgebra duality}
\label{section:duality}
The purpose of this section is to build a framework to reason about an algebra-coalgebra duality akin to Fact \ref{fact:duality_regular} which will allow us, in some cases, to state a general version of Fact~\ref{fact:characterisation_regular_classically_morphisms}. Given a monad  $(S,m,e)$ on an order enriched category $\mathsf{K}$, we first elaborate  more on a functor from the dual of the category $\mathcal{EM}(S)$ to the category $\mathsf{Sat}(S)$.

In what follows, we assume that for the order enriched category $\mathsf{K}$ we have:
\begin{enumerate}[(A)]
\item  a subcategory $\mathsf{J}$  of $\mathsf{K}$ with all objects from $\mathsf{K}$,\label{assumption:1}
\item an identity on objects functor $(-)_-\colon \mathsf{K}\to \mathsf{K}^{op}$ which preserves the order, \ie  $f\leq g\implies f_-\leq g_-$,\label{assumption:2}
\item  for any $f\colon X\to Y\in \mathsf{J}$ the map $f_-\colon Y\to X\in \mathsf{K}$ is its right adjoint in the poset $\mathsf{K}(X,Y)$.\label{assumption:3}
\end{enumerate}
The last item reworded, means that for any $f\colon X\to Y\in \mathsf{J}$ the map $f_-\colon Y\to X$ satisfies $f_-\circ f\geq \mathsf{id}$ and $f\circ f_-\leq \mathsf{id}$. Moreover, we assume that $(S,m,e)$ is a monad on $\mathsf{K}$ such that:
\begin{enumerate}[(A)]
\setcounter{enumi}{3}
\item $S$ is  monotonic, $m_X\colon S^2X\to SX,e_X\colon X\to SX\in \mathsf{J}$ for any object $X$ and $S(f_-) = (Sf)_-$ for any morphism $f\colon X\to Y\in \mathsf{K}$. \label{assumption:4}
\end{enumerate}

\begin{example}
Our prototypical example of $\mathsf{J}$ and $\mathsf{K}$ are $\Set$ and $\kl(\mathcal{P})$ respectively, with the inclusion functor given by $\Set\to \kl(\mathcal{P})$ taking any set to itself and any map $f\colon X\to Y$ to $\{-\}\circ f\colon X\to \mathcal{P}Y;x\mapsto \{f(x)\}$. The order on $\kl(\mathcal{P})$ is defined in a natural manner by $f\colon X\to \mathcal{P}Y\leq g\colon X\to \mathcal{P}Y \iff f(x)\subseteq g(x)$ for any $x\in X$. The category $\kl(\mathcal{P})$ is equipped with a functor $(-)_-\colon \kl(\mathcal{P})\to \kl(\mathcal{P})^{op}$ which assigns to any object itself and to any morphism $f\colon X\to \mathcal{P}Y$ the map $f_-$ given in (\ref{eq:f_op}). It is easy to verify that (\ref{assumption:1})-(\ref{assumption:3})  hold for this choice of categories. Now if we take $S$ to be the lifting $(\overline{\Sigma^\ast},\{-\}\circ m, \{-\}\circ e)$ of the monad $(\Sigma^\ast\times \mathcal{I}d,m,e)$ to $\kl(\mathcal{P})$ then it satisfies (\ref{assumption:4}). In Section \ref{section:examples} we will see other examples of $\mathsf{J}$, $\mathsf{K}$ and $S$ that meet the above requirements.
\end{example}

\begin{wrapfigure}[3]{r}{0.25\textwidth}%
\vspace{-5ex}\hfill\hspace{-2ex}
\begin{tikzpicture}[font=\footnotesize,xscale=1.5,yscale=.8]
\node(1x1) at (0,.5) {$X$};
\node(x1) at (1,1) {$X$};
\node (a1) at (2,1) {$SX$};
\node(fx1) at (1,0) {$SX$};
\node (fa1) at (2,0) {$SSX$};
\draw[->] (x1) -- (a1) node[pos=.5,above] {$a_-$};
\draw[->] (a1) -- (fa1) node[pos=.5,left] {$S(a_-)$};
\draw[->] (x1) -- (fx1) node[pos=.5,left] {$a_-$};
\draw[->] (fx1) -- (fa1) node[pos=.5,below] {$m_-$};
\draw[->] (fx1) -- (1x1) node[pos=.5,below] {$e_-$};
\draw[->] (x1) -- (1x1) node[pos=.5,above] {$\mathsf{id}$};
\end{tikzpicture}
\end{wrapfigure}
\noindent Let us now recall that an Eilenberg-Moore algebra $a\colon S X\to~X$ makes the standard EM-diagrams commute for $(S,m,e)$. By applying $(-)_-$ to these diagrams and by (\ref{assumption:4}) we get commutativity of the diagrams on the right. Hence $S a_- \cdot a_- = m_- \circ a_-$ and   $e_-\circ a_- = \mathsf{id}$. This means that
$m\circ S a_- \circ a_- =m\circ m_-\circ a_-$ together with $e \circ e_-\circ a_- = e.$ By (\ref{assumption:3}) and (\ref{assumption:4}), this finally means that:
$
m \circ S a_- \circ a_- \leq   a_- \text{ and }  e\leq a_-
$.
These are the inequalities that define a saturated $S$-coalgebra. Moreover, if $h\colon X\to Y$ is a homomorphism in $\mathcal{EM}(S)$ between algebras $a\colon SX\to X$ and $b\colon SY\to Y$ then $a_- \circ h_- =  Sh_- \circ b_-$. Hence, $h_-\colon Y\to X$ is a morphism from $b_-$ to $a_-$ in $\mathsf{Sat}(S)$. The above remark allows us to define a functor
$
\mathsf{CoAlg}\colon \mathcal{EM}(S)^{op}\to \mathsf{Sat}(S),
$
which assigns to any algebra $a\colon SX\to X$ the coalgebra $a_-\colon X\to SX$ and to an algebra homomorphism $h\colon X\to Y$ from $a\colon SX\to X$ to $b\colon SY\to Y$ the $S$-coalgebra map $h_-\colon Y\to X$.

\begin{remark}\label{remark:strict_homo_preservation}
The functor $\mathsf{CoAlg}$ maps a homomorphism $h$ between algebras $a$ and $b$ onto a \emph{strict} homomorphism $h_-$ between coalgebras $b_-$ and $a_-$.
\end{remark}

\begin{theorem}\label{lemma:preimage_homomorphism}
Let $b\colon SY\to Y$ be an Eilenberg-Moore. The morphism  $h_-\colon Y\to SX$ is the opposite of a homomorphism $h\colon SX\to Y$ from the Eilenberg-Moore algebra $m_X\colon S^2X\to SX$ to $b\colon SY\to Y$ iff $h_-= Sf\circ b_-$ for some $f\colon Y\to X$.
\end{theorem}

Theorem \ref{lemma:preimage_homomorphism} is a generalisation of \Cref{fact:duality_regular} and provides us with the foundation for generalising the notion of regular maps for non-deterministic automata.

\subsection{Duality}\label{subsection:the_duality}
Let us now denote by $\mathcal{SAT}(S)$ a subcategory of $\mathsf{Sat}(S)$ consisting of saturated $S$-coalgebras whose duals are Eilenberg-Moore algebras and strict homomorphisms between them. By Remark \ref{remark:strict_homo_preservation} we have a category isomorphism
\begin{align}
\mathcal{EM}(S)^{op}\cong \mathcal{SAT}(S) \tag{DUAL}\label{equi:duality}
\end{align}
Note that in the above duality we do not have to assume that the monad $S$ admits saturation. However, if it does then sometimes it is possible to describe members of $\mathcal{SAT}(S)$ (and hence also of $\mathcal{EM}(S)$) in terms of $\alpha^\ast\colon X\to SX$ for a \emph{certain} choice of maps $\alpha\colon X\to SX$. One example of this phenomenon is described below, where for a free monad $F^\ast$ over a functor $F$ the class of objects of $\mathcal{SAT}(F^\ast)$ is (modulo some additional requirements) is given by saturating $F$-coalgebras only.

\subparagraph{Duality for free monads}
Let $F\colon \mathsf{K}\to \mathsf{K}$ be a functor and let $(F^\ast,m,e)$ be the free monad over $F$ together with the transformation $\nu\colon F\implies F^\ast$. Assume that (\ref{assumption:1})-(\ref{assumption:4}) hold for $(F^\ast,m,e)$ on $\mathsf{K}$ with $F^\ast$ admitting saturation. If $\mathsf{K}$ has binary coproducts then the object $F^\ast X$ is the carrier of the initial $F(-)+X$-algebra $i_X\colon F F^\ast X + X\to F^\ast X$ \cite{barrwells:ttt}.  Any Eilenberg-Moore algebra $a\colon F^\ast X\to X$ for the monad $F^\ast$ is uniquely determined by $\underline{a}\defeq FX\stackrel{\nu_X}{\to}F^\ast X\stackrel{a}{\to} X$ as $a\colon F^\ast X \to X$ can be recovered from $\underline{a}$ in terms of a unique homomorphism
between the $F(-)+X$-algebras $i_X\colon FF^\ast X + X \to F^\ast X$ and $[\underline{a},id_X]\colon FX+X\to X$. In this case, if we assume the dual of the saturated map $(X\stackrel{\alpha}{\to} FX\stackrel{\nu_X}{\to} F^\ast X)^\ast$ is an Eilenberg-Moore algebra for the monad $F^\ast$ for any $\alpha\colon X\to FX$ and that for any Eilenberg-Moore algebra $a\colon F^\ast X \to X$, the map $a$ is the least EM-algebra satisfying $F^\ast X \stackrel{\nu_-}{\to} FX\stackrel{\underline{a}}{\to} X \leq F^\ast X \stackrel{a}{\to} X$ then we have the following statement.
\begin{proposition}\label{proposition:duality_free_monads}
Any Eilenberg-Moore algebra for  $F^\ast$ is of the form $\left\{(X\stackrel{\alpha}{\to} FX\stackrel{\nu_X}{\to} F^\ast X)^\ast\right\}_-$ for an $F$-coalgebra $\alpha\colon X\to FX$.
\end{proposition}
Hence, we immediately get the isomorphism between
$\mathcal{EM}(F^\ast)^{op}$ and the subcategory of $\mathsf{Sat}(F^\ast)$ consisting of $\alpha^\ast\colon X\to F^\ast X$ for $\alpha\colon X\to FX \stackrel{\nu_X}{\to} F^\ast X$ as objects and strict (coalgebra) homomorphisms as morphisms. \begin{example}\label{example:duality_classical_automata_free_monad}
The $\Set$-endofunctor $\Sigma^\ast\times \mathcal{I}d$ is a free monad over  $\Sigma\times \mathcal{I}d\colon \Set\to \Set$ with the canonical embedding transformation $\Sigma\times \mathcal{I}d\implies \Sigma^\ast \times \mathcal{I}d$.  As shown in \eg \cite{brengos2015:lmcs} this means that the lifting $\overline{\Sigma^\ast}$ to $\kl(\mathcal{P})$ from Example~\ref{example:monoid_monad} is a free monad over the lifting of the functor $\Sigma\times \mathcal{I}d$  to $\kl(\mathcal{P})$. Moreover, it is not hard to verify that $\overline{\Sigma^\ast}$ satisfies the requirements of this subsection. Hence, Proposition \ref{proposition:duality_free_monads} holds. This precisely means that every Eilenberg-Moore algebra for the monad $\overline{\Sigma^\ast}$ is obtained by taking the duals of saturations of maps of the form $X\to \mathcal{P}(\Sigma\times X)\hookrightarrow \mathcal{P}(\Sigma^\ast \times X)$.
\end{example}

\subsection{Regular behaviours and two modes of recognition}\label{section:two_modes}
The purpose of  this subsection is to generalize the notion of regular language for classical non-deterministic automata from Section \ref{section:classical_regular_revisited} to our more general setting.  Taking into account Theorem \ref{lemma:preimage_homomorphism} (generalizing Fact \ref{fact:duality_regular}) and Table \ref{table:reg} we obtain what follows.

Assume $(T,\mu,\eta)$ and $(S,m,e)$ are monads on $\Set$ and that $(S,m,e)$ lifts to a monad $(\overline{S},\overline{m},\overline{e})$ on $\kl(T)$ via a distributive law $\lambda\colon ST\implies TS$. This yields a monadic structure on the composition $TS$ such that $\kl(TS)=\kl(\overline{S})$. Moreover, assume that (\ref{assumption:1})-(\ref{assumption:4}) are met for $\mathsf{J}=\Set$, $\mathsf{K}=\kl(T)$ and the monad $(\overline{S},\overline{m},\overline{e})$. Arrows between two objects $X,Y$ in $\kl(T)$ will be denoted as before by $X\rightdcirc Y$, \ie $X\rightdcirc Y = X\to TY$. 
Akin to \cite{DBLP:conf/calco/ColcombetP17,urabe_et_al:LIPIcs:2016:6186,brengos2018:concur}, we model automata (with branching type $T$ and linear type $S$) as follows:

\begin{definition}\label{definition:automaton_generelized}
A \emph{$(T,S)$-automaton} is a pair $(\alpha,\chi)$, where $\alpha\colon n\rightdcirc \overline{S}n = n\to TSn$ is in $\mathcal{SAT}(\overline{S})$ and $\chi\colon n\rightdcirc 1 = n\to T1$ is an arbitrary map in $\kl(T)$. By \emph{behaviour} (or \emph{language}) of a state $i\in n$ in $\mathcal{A}$ we mean map $1\to TS1 = 1\rightdcirc \overline{S}1$ given by:
\[L(\mathcal{A},i) = 1\stackrel{i_n}{\rightdcirc} n\stackrel{\alpha}{\rightdcirc} \overline{S}n\stackrel{\overline{S}\chi}{\rightdcirc} \overline{S}1\]
where the assignment $i_n\colon 1\rightdcirc n=1\to Tn$ maps $1$ onto $\eta(i)$ for the unit $\eta_n\colon n\to Tn$.
\end{definition}

\begin{example}
If $T=\mathcal{P}$ and $S=\Sigma^\ast\times \mathcal{I}d$ then a $(\mathcal{P},\Sigma^\ast\times \mathcal{I}d)$-automaton is a pair $(\alpha\colon n\to \mathcal{P}(\Sigma^\ast \times n), \chi\colon n\to \mathcal{P}1)$ where $\alpha$ is a saturated morphism of a map $n\to \mathcal{P}(\Sigma\times n)$ (\emph{conf.} Example \ref{example:duality_classical_automata_free_monad}) and $\chi$ is uniquely determined by the set $\mathcal{F}=\{i\in n\mid \chi(i) \neq \varnothing\}$. Hence, (up to the fact that we replace the original transition of a classical non-deterministic automaton with its saturated version and we replace the set of terminal states with its characteristic function) we obtain the known non-deterministic automaton. Additionally, by Table \ref{table:reg} the above definition of the language coincides with the classical one. 
\end{example}

We are now ready to introduce the notion of regular behaviour: a map $1\to TS1 = 1\rightdcirc \overline{S}1$ in $\kl(TS)=\kl(\overline{S})$ is \emph{regular} if it is a behaviour of a state in a $(T,S)$-automaton. However, as mentioned in the final remarks of Section \ref{section:classical_regular_revisited}, in the case of non-sequential data we need to be able to cover regular morphisms with more than one variable (see also \cite{brengos2018:concur}). Hence, we introduce the following concept. A morphism $1\rightdcirc \overline{S}p=1\to TS p$ in $\kl(\overline{S})=\kl(TS)$ is called \emph{regular} if it is of the form
\begin{equation}
1\stackrel{i_n}{\rightdcirc } n \stackrel{\alpha} {\rightdcirc } \overline{S}n \stackrel{\overline{S} \chi}{\rightdcirc }\overline{S}p.\tag{REG} \label{exp:regular}
\end{equation}
for a $\mathcal{SAT}(\overline{S})$-coalgebra $\alpha$  and  $\chi\colon n\rightdcirc p=n\to Tp$. By Th. \ref{lemma:preimage_homomorphism} the map $ n \stackrel{\alpha} {\rightdcirc } \overline{S}n \stackrel{\overline{S} \chi}{\rightdcirc }\overline{S}p$ is the dual to an Eilenberg-Moore algebra homomorphism from (the free $\overline{S}$-algebra at $p$) $\overline{m}_p\colon \overline{S}^2 p \rightdcirc \overline{S}p$ to (the dual of the saturated coalgebra $\alpha$) $\alpha_-\colon \overline{S}n\rightdcirc n$.

The above definition of regular maps (\ref{exp:regular}) easily extends to morphisms $p\to TSq$ coordinate-wise. 

\subparagraph{Theory of regular behaviours}
As it turns out below (given some extra assumptions) the family of regular maps $p\to TSq$ contains the unit of the monad $TS$, is closed under cotupling and $\kl(TS)$-composition (and hence forms a subtheory of the theory $\mathbb{T}_{TS}$).  
We assume that 
\begin{enumerate}[(I)]
\item For any $n<\omega$ the map $\overline{e}_n\colon  n\rightdcirc \overline{S}n$ is  $\mathcal{SAT}(\overline{S})$,\label{assumption:reg2}
\item for any two $\alpha\colon m\rightdcirc \overline{S}m$ and $\beta\colon n\rightdcirc \overline{S}n$ which are members of $\mathcal{SAT}(\overline{S})$ the map $m+n\stackrel{\alpha+\beta}{\rightdcirc}\overline{S}m+\overline{S}n\stackrel{[\overline{S}\mathsf{inl},\overline{S}\mathsf{inr}]}{\rightdcirc} \overline{S}(m+n)$ is in $\mathcal{SAT}( \overline{S})$\label{assumption:reg2.5}, 
\item if $f\colon \overline{S}k\rightdcirc n$ and $g\colon \overline{S}q\rightdcirc k$ are EM-homomorphisms from $\overline{m}_k\colon \overline{S}^2 k\rightdcirc \overline{S}k$ to $a\colon \overline{S}n\rightdcirc n$ and from $\overline{m}_q\colon \overline{S}^2 q \rightdcirc \overline{S}q$ to $b\colon \overline{S}k\rightdcirc k$ respectively then there is an Eilenberg-Moore algebra $c\colon \overline{S}l\rightdcirc l$ for the monad $\overline{S}$, a  homomorphism $h\colon \overline{S}q\rightdcirc l$ from $\overline{m}_q\colon \overline{S}^2q\rightdcirc \overline{S}q$ to $c\colon \overline{S}l\rightdcirc l$ and $j\in l$ s.t.: $\overline{S}q\stackrel{h}{\rightdcirc} l \stackrel{(j_l)_-}{\rightdcirc}1 = \overline{S}q\stackrel{(\overline{m}_q)_-}{\rightdcirc} \overline{S}^2q\stackrel{\overline{S}g}{\rightdcirc }\overline{S}k\stackrel{f}{\rightdcirc }n\stackrel{(i_n)_-}{\rightdcirc}1.$
\label{assumption:reg3}
\end{enumerate}

\begin{theorem}\label{theorem:regular_maps_form_subtheory}
The collection of objects $p$ for $p<\omega$ and regular maps $p\to TSq$ as morphisms form $p$ to $q$ forms a subtheory of the theory $\mathbb{T}_{TS}$.
\end{theorem}
As a direct corollary of the above theorem we get a Kleene-like theorem characterisation of the regular map theory. 
Indeed, the subtheory of $\mathbb{T}_{TS}$ consisting of regular maps as morphisms is the smallest subtheory which contains all maps of the form $p \rightdcirc q\stackrel{\overline{e}_q}{\rightdcirc }\overline{S}q$ ($\dagger$) and all duals to Eilenberg-Moore algebras $\overline{S}n\rightdcirc n$.  So, if there is a class $\mathcal{B}$ of regular morphisms in $\mathbb{T}_{TS}$ with a common domain and codomain for which $\bigcup_{r\in \mathcal{B}} r^{\mathsf{sat}}$ contains \emph{all} members of $\mathcal{SAT}(\overline{S})$ over a finite state space, where $r^{\mathsf{sat}} = \{ \alpha \mid r\leq \alpha \text{ and }\alpha \in \mathcal{SAT}(\overline{S}) \}$, then the theory of regular morphisms is the smallest subtheory containing all maps from $\mathcal{B}$, all maps ($\dagger$) and being closed under $(-)^\mathsf{sat}$.
 
\begin{example}
The above statement is true for our running example of regular maps for $\mathcal{P}(\Sigma^\ast\times \mathcal{I}d)$. In this case the theory of regular maps is given as the smallest theory containing all maps $m\to \mathcal{P}(\Sigma\times n)$ and being closed under finite unions and Kleene star closure (i.e. saturation) \cite{brengos2018:concur,esik2011,esik2013}. 
\end{example}

\subparagraph{Lawvere theory morphism recognition}
Finally, we point out that in the case when $T=\mathcal{P}$ a natural algebraic characterisation of regular maps holds. Indeed, if the monad $S$ is finitary\footnote{A $\Set$-based monad is called \emph{finitary} if for any $X$ and $x\in SX$ there is a finite subset $X_0\subseteq X$ such that $x\in SX_0$. This assumption about $S$ is technical and related to the fact that in this case the theory $\mathbb{T}_S$ associated with $S$ uniquely determines the  monad $S$ \cite{hyland:power:2007}.} then we can characterize regular morphisms $1\rightdcirc \overline{S}p= 1\to \mathcal{P}Sp$ in terms of preimages of Lawvere theory morphisms as follows.

\begin{definition}
A subset $L\subseteq \mathbb{T}_S(1,p)$ is \emph{recognizable} if there is a Lawvere theory morphism $h\colon\mathbb{T}_S\to \mathbb{T}'$ whose target is finitary, and a subset $T'\subseteq \mathbb{T}'(1,p)$ s.t.~$L=h^{-1}(T')$. 
\end{definition}

\begin{theorem}\label{lemma:from_regular_to_morphisms} \label{theorem:main} The map $L\colon 1\rightdcirc \overline{S} p=1\to \mathcal{P}Sp\in \mathbb{T}_{\mathcal{P}S}(1,p)$ is regular if and only if the set $\{f\colon 1\to Sp\in \mathbb{T}_S(1,p) \mid f(1) \subseteq L(1)\}$ is recognizable.
\end{theorem}

\section{Beyond non-deterministic automata}\label{section:examples}
In this section we illustrate the generality of our results by listing some representative examples of models fitting our framework besides our running example of classical non-deterministic automata and their languages (details are in \cref{section:appendix_examples}).

\subparagraph{Tree automata}
\looseness=-1
For a non-empty set $\Sigma$, let $T_\Sigma$ be the free monad for the endofunctor $\mathcal{I}d\times \Sigma\times \mathcal{I}d$ over \Set:
$T_\Sigma X$ is the set of binary trees with $\Sigma$-labelled nodes and leaves in the set $X$,
$T_\Sigma f$ is the function that replaces leaves according to the function $f$.
Monadic multiplication is tree grafting, and
monadic unit is the embedding into trees with one leaf (and no internal nodes).
The monad $T_\Sigma$ lifts to a monad $\overline{T_\Sigma}$ on $\kl(\mathcal{P})$ by the unique extension of the distributive law of the functor $\mathcal{I}d\times \Sigma\times \mathcal{I}d$ over the monad $\mathcal{P}$ given by the assignment $(L \times \sigma \times R) \mapsto \{(l,\sigma,r) \mid l \in L, r \in R\}$ \cite{brengos2015:lmcs}. The monad $\overline{T_\Sigma}$ coincides with the free monad for the (canonical) lifting of $\mathcal{I}d\times \Sigma\times \mathcal{I}d$ to $\kl(\mathcal{P})$ \cite{brengos2015:lmcs}. 
Additionally, the monad $\overline{T_\Sigma}$ admits saturation \cite{brengos2018:concur}. Indeed, if $\alpha\colon X\rightdcirc \overline{T_\Sigma}X = X\to \mathcal{P}T_\Sigma X$ then $\alpha^\ast(x)= \bigcup_{n<\omega} (\mathsf{id}\cup \alpha)^n (x)$, where $(\mathsf{id}\cup \alpha)^0(x) = \{x\}$, $(\mathsf{id}\cup \alpha)^1(x)=\{x\}\cup \alpha(x)$ and if $t\in (\mathsf{id}\cup \alpha)(x)$ is with leaves in $\{x_1,\ldots,x_k\}$ then a tree which is obtained from $t$ by replacing any occurrence of $x_i$ by some tree from $(\mathsf{id}\cup \alpha)^{n}(x_i)$ is in $(\mathsf{id}\cup \alpha)^{n+1}(x)$. It can be checked that $\overline{T_\Sigma}$ satisfies the requirements of \cref{subsection:the_duality} and hence that \cref{proposition:duality_free_monads} holds. As a consequence, Eilenberg-Moore algebras for $\overline{T_\Sigma}$ are dual to the saturations of morphism of form  $X\to \mathcal{P}(X\times \Sigma \times X)\hookrightarrow \mathcal{P}T_\Sigma X$.

Let $\mathcal{A}=(\alpha\colon n\to \mathcal{P}T_\Sigma n, \chi\colon n\to \mathcal{P}1)$ be a $(\mathcal{P},T_\Sigma)$-automaton. It follows from the above remark and \cref{subsection:the_duality} that the transition map $\alpha \in \mathcal{SAT}(\overline{T_\Sigma})$ of $\mathcal{A}$ is equivalent a map $\hat{\alpha}\colon n\to \mathcal{P}(n\times \Sigma \times n)$ (\ie $\alpha = (\nu_n \circ \hat{\alpha})^\ast$) and that the language accepted by a state $i \in n$ is characterised as:
$
  L(\mathcal{A},i) = 1\rightdcirc n \stackrel{(\nu_X \circ \hat{\alpha})^\ast}{\rightdcirc } \overline{T_\Sigma}n \stackrel{\overline{T_\Sigma}\chi}{\rightdcirc }\overline{T_\Sigma}1,
$
where $\nu$ is the canonical embedding of $\mathcal{P}(\mathcal{I}d \times \Sigma \times \mathcal{I}d)$ into $\mathcal{P}T_\Sigma$.
By comparing this observation with the classical definition of a tree automaton  and its language\footnote{A \emph{tree automaton} is a pair $(\hat{\alpha}\colon n\to \mathcal{P}(n\times \Sigma \times n),\mathcal{F}\subseteq n)$. Intuitively, the language of a state $i$ in a tree automaton is given by the set of finite tree traces whose leaves are all in $\mathcal{F}$. More formally, this concept can be put into our setting by translating the classical definition of the behaviour accepted by the state $i\in n$ into the language of the Kleisli category for the monad $\mathcal{P}$ and is the following: $1\stackrel{i_n}{\rightdcirc} n\stackrel{\hat{\alpha}^\ast}{\rightdcirc} \overline{T_\Sigma}n\stackrel{\overline{T_\Sigma}\chi_\mathcal{F}}{\rightdcirc} \overline{T_\Sigma}1$, where $\chi_\mathcal{F}\colon n\rightdcirc 1=n\to \mathcal{P}1$ is as in \cref{section:classical_regular_revisited}. At this point it is also worth noting that originally non-deterministic tree automata did not have an evident coalgebraic transition map. However, our approach to defining these objects is equivalent to the original. See \eg \cite{pin:automata,brengos2018:concur} for details.} (see \eg \cite{pin:automata}) we immediately get the following correspondence.

\begin{theorem}
Regular maps $1\to \mathcal{P}T_\Sigma 1$ for $\mathcal{P}T_\Sigma$ coincide with languages recognised by tree automata in the classical sense.
\end{theorem}

Theorem  \ref{lemma:from_regular_to_morphisms} instantly provides us with an algebraic characterisation of regular languages for tree automata we will now instantiate on a simple example. Let $\Sigma = \{a,b\}$ and consider a two state tree automaton whose transition morphism $\alpha\colon 2\to \mathcal{P}(2\times\Sigma \times 2)$ is defined by $1\mapsto \varnothing,  2\mapsto \{(2,a,2),(1,b,1)\}$
and $\mathcal{F}=\{1\}$. The language accepted by the state $2$ consists of binary trees of height $>0$ whose nodes preceding the leaf nodes are all $b$ and the remaining non-leaf nodes are $a$. By following the guidelines of the proof of Theorem \ref{lemma:from_regular_to_morphisms} we build a theory morphism $h\colon \mathbb{T}_{T_\Sigma}\to \mathbb{T}'$ which recognizes the language of the state $2$ in the above automaton. The finitary theory $\mathbb{T}'$ and the theory morphism $h$ are determined by the automaton $(\alpha,\mathcal{F})$ with the hom-set $\mathbb{T}'(1,1)$ of  $\mathbb{T}'$ consisting of $4$ elements which are assignments  $\mathcal{P}2\to \mathcal{P}2$ given in the table below. The composition in $\mathbb{T}'$ of morphisms $\mathbb{T}'(1,1)$ is the ordinary assignment composition of maps $\mathcal{P}2\to \mathcal{P}2$ given in the reversed order.
\begin{wrapfigure}[4]{l}{0.42\textwidth}
\vspace{-0.4cm}
\resizebox{0.42\textwidth}{!}{
$\begin{array}{c|c|c|c|c|}
 & \mathsf{id} & a_A=a_A^2 & b_A=a_A\circ b_A  & b_A^2=b_A\circ a_A\\\hline
\varnothing &\varnothing &\varnothing  &\varnothing &\varnothing\\ \hline
\{1\} & \{1\} & \varnothing & \{2\} &  \varnothing  \\ \hline
\{2\} & \{2\} & \{2\} & \varnothing  & \varnothing  \\ \hline
\{1,2\} & \{1,2\} & \{2\} & \{2\} &\varnothing  \\ \hline
\end{array}$
}
\end{wrapfigure}
The morphism $h\colon \mathbb{T}_{T_\Sigma}\to \mathbb{T}'$, if restricted to $\mathbb{T}_{T_\Sigma}(1,1)$, maps the tree $1$ to $\mathsf{id}$, any tree $t$ such that after the composition with the tree $(1,b,1)$ in $\mathbb{T}_{T_\Sigma}$ the result is in the language of the state $2$ to $a_A$, any tree from the language accepted by the state $2$ onto $b_A$ and any other to $b^2_A$. We see that our language of the state $2$ is given in terms of $h^{-1}(\{b_A\})$.

\subparagraph{Weighted automata}
For $(\mathcal{S},+,0,\cdot,1,\leq)$ a positive $\omega$-semiring\footnote{
A positive $\omega$-semiring is relational structure $(\mathcal{S},+,0,\cdot,1,\leq)$ with the property that 
$(\mathcal{S},+,0,\cdot,1)$ is a semiring with countable sums,
$(\mathcal{S},\leq,0)$ is an $\omega$-complete partial order,
$+$ and $\cdot$ are $\omega$-continuous in both arguments.
}
(\eg the set of non-negative reals extended with positive infinity $[0,+\infty]$), an $\mathcal{S}$-multiset is a pair $(X,\phi)$ where $X$ is a set and $\phi\colon X \to \mathcal{S}$ is a function such that the set $\supp(\phi) = \{x \mid \phi(x) > 0 \}$ (called \emph{support}) is countable. We write $\mathcal{S}$-multiset as formal sums. The $\mathcal{S}$-multiset functor $\mathcal{M_S}\colon \Set \to \Set$ assigns to every set $X$ the set $\mathcal{M_S}X$ of $\mathcal{S}$-multisets with universe $X$, and to every function $f\colon X \to Y$ the function mapping each $(X,\phi)$ to $(X, \sum_{x \in \supp(\phi)} \phi(x) \bullet f(x))$. 
This functor carries a monad structure $(\mathcal{M_S},\mu,\eta)$ whose multiplication $\mu$ and unit $\eta$ are given on each set $X$ by the mappings: $(\mathcal{M_S}X,\psi) \mapsto \left(X,\sum (\phi(x) \cdot \psi(\phi)) \bullet x\right)$ and $x \mapsto (X,x \bullet 1)$.
The probability distribution monad $\mathcal{D}$ is a submonad of $\mathcal{M}_{[0,\infty]}$ \cite{brengos2015:jlamp}.

\looseness=-1
The free monad $\Sigma^\ast$ lifts to a monad $\overline{\Sigma^\ast}$ on $\kl(\mathcal{M_S})$ by the unique extension of the distributive law of the functor $\Sigma\times \mathcal{I}d$ over the monad $\mathcal{M_S}$ given by the assignment $\Sigma\times (X,\phi) \mapsto (\Sigma \times X, \sum\phi(x)\bullet(\sigma,x))$ \cite{brengos2015:jlamp}.
The main problem with this choice of monads is that they do not fit our setting from \cref{section:duality} directly. Indeed, $\kl(\mathcal{M_S})$ is not self-dual (due to  the limited size of the cardinality of the support of functions in $\mathcal{M_S}$). There are two workarounds to this problem: one is to extend the definition of $\mathcal{M_S}$ to cover functions of arbitrary supports, the other is to realize that when dealing with regular behaviours as in (\ref{exp:regular}) we actually focus on systems over a finite state space. 
Hence, we restrict w.l.o.g.~to the subcategory identified by finite sets. In this case, if $\Sigma^\ast$ is countable then any Eilenberg-Moore algebra $a\colon \overline{\Sigma^\ast} n \rightdcirc n=\Sigma^\ast\times n \to \mathcal{M_S}n$ on $\kl(\mathcal{M_S})$ yields a saturated system $a_-\colon  n\rightdcirc \overline{\Sigma^\ast} n = n\to \mathcal{M_S}({\Sigma^\ast \times n})$ by simple currying and uncurrying. Since $\overline{\Sigma^\ast}$ is the free monad over $\overline{\Sigma}$ on $\kl(\mathcal{M_S})$ \cite{brengos2015:lmcs}, the Eilenberg-Moore algebra $a$ is uniquely determined by $\underline{a}\colon \Sigma\times n \to \mathcal{M_S}n$ (\emph{conf} Subsec. \ref{subsection:the_duality}). Hence, so is its dual $a_-$.

Let $\mathcal{A} = (\alpha\colon n\to \mathcal{M_S}(\Sigma^\ast \times n), \chi\colon n\to \mathcal{M_S}1)$ be a $(\mathcal{M_S},\Sigma^\ast)$-automaton with the behaviour of a state $i\in n$ given by 
\vspace{-.4ex}\[\vspace{-.4ex}
L(\mathcal{A},i) = 1\rightdcirc n \stackrel{a}{\rightdcirc } \overline{\Sigma^\ast}n \stackrel{\overline{\Sigma^\ast}\chi}{\rightdcirc }\overline{\Sigma^\ast}1\]
where $\nu$ is the canonical embedding of $\mathcal{M_S}(\Sigma\times \mathcal{I}d)$ into $\mathcal{M_S}(\Sigma^\ast\times \mathcal{I}d)$.
This is essentially the classical presentation of an automaton weighted over $\mathcal{S}$\footnote{A \emph{weighted automaton} is a pair $(\hat{\alpha}\colon n\to \mathcal{M_S}(\Sigma \times n),\chi\colon n \to \mathcal{S})$ and the language of a state $i \in n$ is the $\mathcal{S}$-multisetset $\left(\Sigma^\ast,\lambda\sigma_1\dots\sigma_k. \sum \left\{\chi(j_k)\cdot \prod_{p< k} \hat\alpha(j_{p})(\sigma_{p+1},j_{p+1}) \,\middle|\, j_0 = i, j_1,\dots,j_k \in n \right\}\right)$.} \cite{DBLP:journals/iandc/Schutzenberger61b,DBLP:journals/iandc/BonchiBBRS12}.

\begin{theorem}
Regular maps $1\rightdcirc\overline{\Sigma^\ast} 1 =1 \to  \mathcal{M}_\mathcal{S}(\Sigma^\ast \times 1)$ for $\mathcal{M}_\mathcal{S}({\Sigma^\ast}\times \mathcal{I}d)$ coincide with languages recognised by weighted automata in the classical sense.
\end{theorem}

Regular morphisms form a subtheory of $\mathbb{T}_{\mathcal{M}_\mathcal{S}({\Sigma^\ast}\times \mathcal{I}d)}$, as weighted languages enjoy Kleene Theorem \cite{DBLP:journals/iandc/SilvaBBR11}.

Weighted tree automata and their languages are captured by our framework as well since the monad $T_\Sigma$ lifts to $\kl(\mathcal{M_S})$.

\subparagraph{Fuzzy automata} 
\looseness=-1
For $(\mathcal{Q},\cdot ,1,\leq)$ a unital quantale (\eg the real unit interval $([0,1],\cdot,1,\leq)$), let $(\mathcal{P_Q},\bigcup_\mathcal{Q},\{-\}_\mathcal{Q})$ be the \emph{$\mathcal{Q}$-fuzzy powerset monad}
and observe that the free monad $\Sigma^\ast$ lifts to $\kl(\mathcal{P_Q})$ via the unique extension of the distributive law $\lambda\colon \Sigma\times \mathcal{I}d \to \mathcal{P_Q}$ given by $(\Sigma \times (X,\phi)) \mapsto (\Sigma \times X, \lambda (\sigma,x).\phi(x))$,
\Cref{proposition:duality_free_monads} holds for $\overline{\Sigma^\ast}$ since $\kl(\mathcal{P_Q})$ admits saturation \cite{brengos2019:lmcs} and meets the requirements of \cref{subsection:the_duality}. As a consequence,
Eilenberg-Moore algebras for $\overline{\Sigma^\ast}$ are dual to the saturations of morphism of the form  $X\to \mathcal{P_Q}(\Sigma \times X)\hookrightarrow \mathcal{P_Q}(\Sigma^\ast \times X)$.

Let $\mathcal{A} = (\alpha\colon n\to \mathcal{P_Q}(\Sigma^\ast \times n), \chi\colon n\to \mathcal{P_Q}1)$ be a $(\mathcal{P_Q},\Sigma^\ast)$-automaton.
The transition map $\alpha \in \mathcal{SAT}(\overline{\Sigma^\ast})$ of $\mathcal{A}$ is equivalently defined (by saturation and construction of $\Sigma^\ast$) as a map $\hat{\alpha}\colon n\to \mathcal{P_Q}(\Sigma \times n)$ and the language accepted by a state $i \in n$ as 
$
  L(\mathcal{A},i) = 1\rightdcirc n \stackrel{(\nu_X \circ \hat{\alpha})^\ast}{\rightdcirc } \overline{\Sigma^\ast}n \stackrel{\overline{\Sigma^\ast}\chi}{\rightdcirc }\overline{\Sigma^\ast}1
$. If we now recall the classical definition of a fuzzy automaton and its language\footnote{A \emph{fuzzy automaton} is a pair $(\hat{\alpha}\colon n\to \mathcal{P_Q}(\Sigma \times n),(n,\chi))$ and the language of a state $i \in n$ is the fuzzy set $\left(\Sigma^\ast,\lambda \sigma_1\dots\sigma_k. \bigvee \left\{\chi(j_k)\cdot \prod_{p< k} \hat\alpha(j_{p})(\sigma_{p+1},j_{p+1}) \,\middle|\, j_0 = i, j_1,\dots,j_k \in n \right\}\right)$.}
(\eg from \cite{mordeson:2002fuzzy,doostfatemeh:2005fuzzy}) we immediately get the following coincidence.

\begin{theorem}
Regular maps $1\rightdcirc\overline{\Sigma^\ast} 1 =1 \to  \mathcal{P_Q}(\Sigma^\ast \times 1)$ for $\mathcal{P_Q}({\Sigma^\ast}\times \mathcal{I}d)$ coincide with languages recognised by fuzzy automata in the classical sense.
\end{theorem}

Regular maps form a subtheory of $\mathbb{T}_{\mathcal{P_Q}({\Sigma^\ast}\times \mathcal{I}d)}$, as fuzzy languages enjoy Kleene Theorem as they are closed under composition, sums, and saturation \cite{mordeson:2002fuzzy,DBLP:journals/fss/StamenkovicC12}.

Fuzzy tree automata and their languages are similarly captured by our framework since $T_\Sigma$ lifts to $\kl(\mathcal{P_Q})$ and the resulting monad meets the hypotheses of \cref{proposition:duality_free_monads}.

\section{Conclusion} \label{section:summary}

The paper's goal was to build a connection between two approaches towards categorical language theory: the coalgebraic and algebraic language theory for monads. For a pair of $\mathsf{Set}$-based monads $T$ and $S$ (with $T$ modelling the branching type and $S$ the linear type) which admit a monadic structure on their composition $TS$ we defined regular maps $p\to TSq$ that generalize regular languages known in classical non-deterministic automata theory. Although these maps are of coalgebraic nature (as they are, roughly speaking, behaviours of finite $(T,S)$-automata), they arise as duals of certain Eilenberg-Moore algebra homomorphisms. The key ingredient to all our results was the Eilenberg-Moore algebra-saturated coalgebra duality based on the self-duality of  (a certain subcategory of) $\kl(T)$.

We showed that, given some extra assumptions, regular maps form a subtheory of the Lawvere theory associated with the monad $TS$. Moreover, we stated a Kleene-like theorem saying that the Lawvere theory of regular morphisms is the smallest subtheory containing all branching type maps and duals of Eilenberg-Moore algebras of a lifting of $S$ to $\kl(T)$.

Additionally, whenever $T=\mathcal{P}$ we showed that regular maps of type $\mathcal{P}S$ are characterised as maps recognized by Lawvere theory morphisms whose codomains are finitary theories. 

Although, our running example were classical non-deterministic automata and regular languages we instantiated the theory presented in this paper to tree automata, fuzzy automata and weighted automata.

\subparagraph{Related work}
We build on the coalgebraic language theory from \cite{bloomesik:93,urabe_et_al:LIPIcs:2016:6186,esik2011,esik2013,brengos2018:concur}, the algebraic language theory stated in the context of Eilenberg-Moore algebras in \cite{bojanczyk2015} and some classical results from \cite{pin:automata,Wilke93,eilenberg}. We are unaware of any research which exploits the Eilenberg-Moore algebra-saturated coalgebra duality on the categorical level to show an equivalence between regular and recognizable languages akin to our approach. The closest are \cite{DBLP:journals/tocl/AdamekMMU19} and \cite{DBLP:journals/corr/Salamanca17}, both study Eilenberg-type dualities:
The first work characterises deterministic word-automata in a locally finite variety whereas our work applies also \eg to tree automata.
The second work provides a duality result between algebras for a monad and coalgebras for a comonad whereas we investigate dualities between algebras and saturated coalgebras for the same type. Moreover, we present a Kleene-like theorem for regular maps which (up to our knowledge) has not been stated at this level of generality.

\subparagraph{Future work}
This paper provides evidence that Lawvere theory morphism recognition is a natural context to investigate whether concepts and properties known in the algebraic language theory (syntactic algebra, star-free language characterization, and more \cf \cite{weil2004}) can be stated at this level of generality.
A key result in this direction would be the extension of the notion of morphism recognition and our results beyond non-determinism ($T = \mathcal{P}$). 
We see the recent enriched view on the extended finitary monad-Lawvere theory correspondence \cite{DBLP:journals/lmcs/GarnerP18} as a helpful stepping stone towards this goal.

\bibliography{biblio}
\clearpage

\appendix
\section{Basic notions (extended)}\label{section:appendix_basic_notions}

\subsection{Algebras and coalgebras}
Let $F:\mathsf{C}\to \mathsf{C}$ be a functor. An \emph{$F$-coalgebra} (\emph{$F$-algebra}) is a morphism $\alpha:A\to FA$ (resp. $a:FA\to A$). The object $A$ is called a \emph{carrier} of the underlying $F$-(co)algebra. Given two coalgebras $\alpha:A\to FA$ and $\beta:B\to FB$ a morphism $h:A\to B$ is \emph{homomorphism} from $\alpha$ to $\beta$ provided that $\beta\circ h= F(h)\circ \alpha$. For two algebras $a:FA\to A$ and $b:FB\to B$ a morphism $h:A\to B$ is called \emph{homomorphism} from $a$ to $b$ if $b\circ F(h) = h\circ a$. The category of all $F$-coalgebras ($F$-algebras) and homomorphisms between them is denoted by $\mathsf{CoAlg}(F)$ (resp. $\mathsf{Alg}(F)$).
Let $\Sigma$ be a set of \emph{labels}.

\subsection{Monads}
A \emph{monad} on $\mathsf{C}$ is a triple $(T,\mu,\eta)$, where $T:\mathsf{C}\to \mathsf{C}$ is an endofunctor and $\mu:T^2\implies T$, $\eta:\mathcal{I}d\implies T$ are two natural transformations for which the following diagrams commute:
$$
\begin{tikzpicture}
\node (t2) at (0,0) {$T^2$};
\node (t) at (2,0) {$T$};
\node (t3) at (0,1) {$T^3$};
\node (t22) at (2,1) {$T^2$};

\draw[->] (t3) -- (t22) node[pos=.5,above] {$\mu$};
\draw[->] (t3) -- (t2) node[pos=.5,left] {$T\mu$ };
\draw[->] (t2) -- (t) node[pos=.5,below] {$\mu$};
\draw[->] (t22) -- (t) node[pos=.5,right] {$\mu$};

\node (at2) at (4,0) {$T^2$};
\node (at) at (6,0) {$T$};
\node (at3) at (4,1) {$T$};
\node (at22) at (6,1) {$T^2$};

\draw[->] (at3) -- (at22) node[pos=.5,above] {$T\eta$};
\draw[->] (at3) -- (at2) node[pos=.5,left] {$\eta_T$ };
\draw[->] (at3) -- (at) node[pos=.3,below] {$\mathsf{id}$ };
\draw[->] (at2) -- (at) node[pos=.5,below] {$\mu$};
\draw[->] (at22) -- (at) node[pos=.5,right] {$\mu$};
\end{tikzpicture}
$$
The transformation $\mu$ is called  \emph{multiplication} and $\eta$ \emph{unit}.

\subsection{Kleisli category} We start this section by recalling the notion of Kleisli category for a monad and listing basic examples of monads and their Kleisli categories we will work with throughout this paper.

Any monad $(T:\mathsf{C}\to \mathsf{C},\mu,\eta)$ gives rise to the \emph{Klesli category} $\kl(T)$ for $T$: it has the class of objects equal to the class of objects of $\mathsf{C}$ and for two objects $X,Y$ in $\kl(T)$ we have ${\kl(T)}(X,Y) = {\mathsf{C}}(X,TY)$
with the composition $\bullet$ in $\kl(T)$ defined between two morphisms $f:X\to TY$ and $g:Y\to TZ$ by
$g\bullet f := \mu_Z \circ T(g) \circ f$. In order to emphasize the distinction between morphisms in $\mathsf{C}$ and in $\kl(T)$ any morphism between two objects $X,Y$ will be denoted by $X\to Y$ if it is a morphism in $\mathsf{C}$ and $X\rightdcirc Y$ if it is a morphism in $\kl(T)$. Hence, $X\rightdcirc Y = X\to TY$ and $X\stackrel{f}{\rightdcirc } Y \stackrel{g}{\rightdcirc }Z =X\stackrel{f}{\to } TY \stackrel{Tg}{\to }T^2Z\stackrel{\mu_Z}{\to} TZ.$
 The category $\mathsf{C}$ is a subcategory of $\kl(T)$ where the inclusion functor $(-)^{\sharp}$ sends each object $X\in \mathsf{C}$ to itself and each map $f:X\to Y$ in $\mathsf{C}$ to the morphism
$
f^{\sharp}:X\to TY; f^{\sharp} \defeq \eta_Y\circ f.
$

\begin{example}\label{example:powerset_monad_kleisli}
The powerset endofunctor $\mathcal{P}:\Set\to \Set$ is a monad whose multiplication $\bigcup:\mathcal{P}^2\implies\mathcal{P}$ and unit $\{-\}:\mathcal{I}d\implies \mathcal{P}$ are given by
$\bigcup:\mathcal{P}\mathcal{P}X\to \mathcal{P}X; S \mapsto \bigcup S$ and $\{-\}_X:X\to \mathcal{P}X; x\mapsto \{x\}$.
The Kleisli category  $\kl(\mathcal{P})$  consists of sets as objects and maps $f:X\to \mathcal{P}Y$ and $g:Y\to \mathcal{P}Z$ with the composition $g\bullet f:X\to \mathcal{P}Z$ defined as follows: $g\bullet f(x) = \{z\in Z \mid z\in \bigcup g(f(x))\}.$ It is a simple exercise to prove that this category is isomorphic to the category $\mathsf{Rel}$ of sets as objects and binary relations with standard relation composition as morphisms and their composition.
\end{example}

\subsubsection{Distributive laws and liftings}

Let $(T,\mu,\eta)$ be a monad on $\mathsf{C}$ and $S$ a $\mathsf{C}$-endofunctor.
A \emph{distributive law}  $\lambda: ST\implies TS$ of the functor $S$ over the monad $T$, \ie natural transformation which satisfies extra conditions (see \eg \cite{mulry:mfps1993} for details):

\resizebox{0.4\textwidth}{!}{
\begin{tikzpicture}
\node (FTX1) at (0,0) {$STX$};
\node (TFX1) at (2,0) {$TSX$};
\node (FX1) at (1,1) {$SX$};

\draw[->] (FX1) -- (FTX1) node[pos=.5,left] {\tiny $S\eta$};
\draw[->] (FX1) -- (TFX1) node[pos=.5,right] {\tiny $\eta$};
\draw[->] (FTX1) -- (TFX1) node[pos=.5,below] {\tiny $\lambda$};
\hspace{4cm}
\node (FTX2) at (0,0) {$STX$};
\node (TFX2) at (4,0) {$TSX$};
\node (FT2X2) at (0,1) {$ST^2 X$};
\node (TFTX2) at (2,1) {$TST X$};
\node (T2FX2) at (4,1) {$T^2SX$};

\draw[->] (FT2X2) -- (FTX2) node[pos=.5,left] {\tiny $S\mu$};
\draw[->] (FTX2) -- (TFX2) node[pos=.5,below] {\tiny $\lambda$};
\draw[->] (T2FX2) -- (TFX2) node[pos=.5,right] {\tiny $\mu$};
\draw[->] (FT2X2) -- (TFTX2) node[pos=.5,above] {\tiny $\lambda$};
\draw[->] (TFTX2) -- (T2FX2) node[pos=.5,above] {\tiny $T\lambda$};
\end{tikzpicture}
}

A \emph{distributive law of the monad}  $(S,m,e)$ on $\mathsf{C}$ \emph{over} the monad $(T,\mu,\eta)$ \cite{barrwells:ttt} is a distributive law $\lambda:ST\implies TS$ of the underlying functor $S$ over the monad $T$ which additionally satisfies:

\resizebox{0.4\textwidth}{!}{
\begin{tikzpicture}
\node (FTX1) at (0,0) {$STX$};
\node (TFX1) at (2,0) {$TSX$};
\node (FX1) at (1,1) {$TX$};

\draw[->] (FX1) -- (FTX1) node[pos=.5,left] {\tiny $e$};
\draw[->] (FX1) -- (TFX1) node[pos=.5,right] {\tiny $Te$};
\draw[->] (FTX1) -- (TFX1) node[pos=.5,below] {\tiny $\lambda$};

\hspace{4cm}
\node (FTX2) at (0,0) {$STX$};
\node (TFX2) at (4,0) {$TSX$};
\node (FT2X2) at (0,1) {$SST X$};
\node (TFTX2) at (2,1) {$STS X$};
\node (T2FX2) at (4,1) {$TSSX$};

\draw[->] (FT2X2) -- (FTX2) node[pos=.5,left] {\tiny $m$};
\draw[->] (FTX2) -- (TFX2) node[pos=.5,below] {\tiny $\lambda$};
\draw[->] (T2FX2) -- (TFX2) node[pos=.5,right] {\tiny $Tm$};
\draw[->] (FT2X2) -- (TFTX2) node[pos=.5,above] {\tiny $S\lambda$};
\draw[->] (TFTX2) -- (T2FX2) node[pos=.5,above] {\tiny $\lambda$};
\end{tikzpicture}
}

Let $\lambda:ST\implies TS$ be a distributive law of a functor $S:\mathsf{C}\to \mathsf{C}$ over a monad $(T,\mu,\eta)$ on $\mathsf{C}$. This allows us to define a functor $\overline{S}:\kl(T)\to \kl(T)$ as follows. Any object $X\in \kl(T)$ is mapped onto $SX\in \kl(T)$. Any morphism $f:X\rightdcirc Y=X\to TY$ is mapped onto $\overline{S}f:\overline{S}X\rightdcirc \overline{S}Y=SX\to TSY$ given by:
$\overline{S}f\defeq SX\stackrel{S f}{\to} STY\stackrel{\lambda_Y}{\to} TSY.$
 We then say that $S:\mathsf{C}\to\mathsf{C}$ \emph{lifts to}   $\overline{S}:\kl(T)\to\kl(T)$ via $\lambda$.

If $\lambda:ST\implies TS$ is a distributive law of a \emph{monad}  $(S,m,e)$ over the monad $(T,\mu,\eta)$ then this allows to introduce a monadic structure $(\overline{S},\overline{m},\overline{e})$ on the lifting $\overline{S}$ of the functor $S$ to $\kl(T)$ by putting $ \overline{e}_X\defeq \eta_X \circ e_X$ and  $\overline{m}_X\defeq \eta_{SX}\circ m_{X}$.
We then say that the monad $(S,m,e)$ on $\mathsf{C}$ \emph{lifts to} the monad $(\overline{S},\overline{m},\overline{e})$ on $\kl(T)$ via $\lambda$.

This also yields a monadic structure on $TS:\mathsf{C}\to \mathsf{C}$ with $\kl(TS) = \kl(\overline{S})$, where the composition $g\cdot f$ is given in $\mathsf{C}$ for $f\colon X\to TSY$ and $g\colon Y\to TSZ$ by:

\begin{center}
\resizebox{0.7\textwidth}{!}{

	\begin{tikzpicture}[auto, xscale=1.8,font=\small,baseline=(current bounding box.center)]
		\node (n0) at (.2,0) {$X$};
		\node (n1) at (1,0) {$TSY$};
		\node (n2) at (2,0) {$(TS)^2Z$};
		\node (n3) at (3,0) {$T^2S^2Z$};
		\node (n5) at (4,0) {$T^2SZ$};
		\node (n6) at (5,0) {$TSZ$};
		\draw[->,rounded corners] (n0) -- +(0,.7) -| node[pos=.25] {$g \cdot f$} (n6);
		\draw[->] (n0) to node[swap] {$f$} (n1);
		\draw[->] (n1) to node[swap] {$TSg$} (n2);
		\draw[->] (n2) to node[swap] {$T\lambda_{SZ}$} (n3);
		\draw[->] (n3) to node[swap] {$T^2 m_Z$} (n5);
		\draw[->] (n5) to node[swap] {$\mu_{SZ}$} (n6);
	\end{tikzpicture}
}
\end{center}

\subsubsection{Free monads and their liftings}
 A \emph{free monad} \cite{barrwells:ttt} over a functor $F:\mathsf{C}\to \mathsf{C}$ is a monad $(F^\ast,m,e)$ together with a natural transformation $\nu:F\implies F^\ast$ such that for any monad $S=(S,m',e')$ on $\mathsf{C}$ and a natural transformation $s:F\to S$ there is a unique monad morphism $\overline{s}:F^\ast\to S$ such that $\overline{s}\circ \nu=s$. The free monad $F^\ast$ over $F$ has an explicit construction in terms of free $F$-algebras as follows.  Let $\mathsf{C}$ admit binary coproducts $+$ with the cotupling denoted by $[-,-]$ and the coprojections into the first and second component of $X+Y$ by $\mathsf{inl}:X\to X+Y$ and $\mathsf{inr}:Y\to X+Y$ respectively. Assume $F$ has an initial $F(-)+X$ algebra (=free $F$-algebra over $X$) for any $X$.   This allows us to define a functor $F^\ast:\mathsf{K}\to \mathsf{K}$ which maps any object $X$ onto the carrier of
the initial $F(-)+X$ algebra $i_X:FF^\ast X+X\to F^\ast X$. Moreover, this functor carries a monadic structure $(F^\ast,m,e)$ which arises from universal properties of the initial algebras $i_X$ and turns $(F^\ast,m,e)$ into a free monad over $F$ with a natural transformation $\nu:F\implies F^\ast$ on its $X$-component given by:
$$
FX\stackrel{F\mathsf{inr}}{\to} F(FF^\ast X + X)\stackrel{F i_X}{\to} FF^\ast X\stackrel{\mathsf{inl}}{\to}FF^\ast X+X \stackrel{i_X}{\to} F^\ast X.
$$

 Additionally, if $F:\mathsf{C}\to\mathsf{C}$ is a funtor that lifts to $\kl(T)$ and admits a free monad $F^\ast:\mathsf{C}\to \mathsf{C}$ then the monad $F^\ast$ lifts to a monad $\overline{F^\ast}:\kl(T)\to \kl(T)$ which is a free monad over the lifting $\overline{F}$ \cite{brengos2015:lmcs}.
\begin{example}
Our two running examples of monads, namely $\Sigma^\ast\times \mathcal{I}d$ and $T_\Sigma$, are both free monad over the functors $\Sigma\times \mathcal{I}d$ and $\mathcal{I}d\times \Sigma\times \mathcal{I}d$ respectively. Hence, since the functors $\Sigma\times \mathcal{I}d$ and $\mathcal{I}d\times \Sigma\times \mathcal{I}d$ lift to $\kl(\mathcal{P})$ \cite{hasuo07:trace}, the monads $\Sigma^\ast\times \mathcal{I}d$ and $T_\Sigma$ also do. Their distributive laws over $\mathcal{P}$ are given in Example \ref{example:kleisli_powerset} and Section \ref{section:examples}. The statement remains true if we change $\mathcal{P}$ into $\mathcal{P_Q}$ or $\mathcal{M_S}$.

\end{example}

\subsection{Eilenberg-Moore algebras}

Given a monad $(T,\mu,\eta)$ on a category $\mathsf{C}$ we say  that a $T$-algebra $a:TX\to X$ is an \emph{Eilenberg-Moore algebra} if $a\circ Ta= a \circ \mu_X$ and $id_X = a\circ \eta_X$:
$$
\begin{tikzpicture}
\node(x1) at (0,0) {$X$};
\node (a1) at (2,0) {$T X$};
\node(fx1) at (0,1) {$T X$};
\node (fa1) at (2,1) {$T^2 X$};
\draw[->] (a1) -- (x1) node[pos=.5,above] {$a$};
\draw[->] (fa1) -- (a1) node[pos=.5,right] {$Ta$};
\draw[->] (fx1) -- (x1) node[pos=.5,left] {$a$};
\draw[->] (fa1) -- (fx1) node[pos=.5,above] {$\mu$};

\node(1x1) at (4,1) {$X$};
\node (1a1) at (6,1) {$T X$};
\node (1fa1) at (6,0) {$X$};
\draw[->] (1a1) -- (1fa1) node[pos=.5,right] {$a$};
\draw[->] (1x1) -- (1a1) node[pos=.5,above] {$\eta$};
\draw[->] (1x1) -- (1fa1) node[pos=.5,below] {$\mathsf{id}$};
\end{tikzpicture}
$$
The collection of all Eilenberg-Moore algebras for the monad $T$ as objects with $T$-algebra homomorphisms as morphisms forms the category of Eilenberg-Moore algebras denoted by $\mathcal{EM}(T)$. For any object $X$ the map $\mu_X:TTX\to TX$ is an Eilenberg-Moore algebra over $X$. Moreover, given an Eilenberg-Moore algebra $a:TA\to A$ and a morphism $h: X\to A$ there is a unique morphism $\overline{h}:TX\to A$ which is a homomorphism from $\mu_X$ to $a$ satisfying $\overline{h} \circ \eta_X = h$. This turns the algebra $\mu_X:T^2 X\to TX$ with $\eta:X\to TX$ into a free Eilenberg-Moore algebra over $X$.

\begin{theorem}\label{theorem:EM_on_kleisli_to_EM_on_base}
Assume $(S,m,e)$ is a monad on $\mathsf{C}$ that lifts to a monad $(\overline{S},\overline{m},\overline{e})$ on the Kleisli category for $(T,\mu,\eta)$ on $\mathsf{C}$ via $\lambda$. This induces a functor $|-|:\mathcal{EM}(\overline{S})\to \mathcal{EM}(S)$ which maps any $\overline{S}$-algebra $a:\overline{S}X\rightdcirc X = SX\to TX$ onto  $$|a|=STX\stackrel{\lambda}{\to}TSX\stackrel{Ta}{\to} TTX \stackrel{\mu}{\to} TX$$ and any algebra homomorphism $h:X\rightdcirc Y=X\to TY$ between the algebras $a:\overline{S}X\rightdcirc X$ and $b:\overline{S}Y\rightdcirc Y$ onto the map $|h| = TX\stackrel{Th}{\to} TTY\stackrel{\mu_Y}{\to} TY.$
\begin{center}
\resizebox{0.6\textwidth}{!}{
\begin{tikzpicture}
\node(sx) at (0,1) {$\overline{S}X$};
\node (x) at (0,0) {$X$};
\node(sy) at (3,1) {$\overline{S} Y$};
\node (y) at (3,0) {$Y$};
\draw[->] (x) -- (y) node[pos=.5,below] {$h$};
\draw[->] (sx) -- (x) node[pos=.5,left] {$a $};
\draw[->] (sy) -- (y) node[pos=.5,right] {$b$};
\draw[->] (sx) -- (sy) node[pos=.5,above] {$\overline{S} h$};
\filldraw[fill=white, draw=black]  (0,0.5) circle (0.08);
\filldraw[fill=white, draw=black]  (3,0.5) circle (0.08);
\filldraw[fill=white, draw=black]  (1.5,0) circle (0.08);
\filldraw[fill=white, draw=black]  (1.5,1) circle (0.08);
\node (ineq) at (3.7,0.5) {$\implies$};
\node(sx1) at (6.5,1) {$STX$};
\node (x1) at (6.5,0) {$TX$};
\node(sy1) at (8.5,1) {$STY$};
\node (y1) at (8.5,0) {$TY$};
\draw[->] (x1) -- (y1) node[pos=.5,below] { $\mu_Y\circ Th$};
\draw[->] (sx1) -- (x1) node[pos=.5,left] {$\mu_X\circ Ta\circ \lambda_X$};
\draw[->] (sy1) -- (y1) node[pos=.5,right] {$|b|$};
\draw[->] (sx1) -- (sy1) node[pos=.5,above] { $S |h|$};
\end{tikzpicture}
}
\end{center}
\end{theorem}

\begin{proof}(Theorem \ref{theorem:EM_on_kleisli_to_EM_on_base})
At first we need to show that if $a:\overline{S}X\rightdcirc X=SX\to TX$ is an Eilenberg-Moore algebra  for $\overline{S}$ then $|a|$ is in $\mathcal{EM}(S)$. This is indeed the case since the following diagrams commute in $\mathsf{C}$:
\begin{center}
\begin{tikzpicture}
\node(STX) at (0,0) {$STX$};
\node(TSX) at (2,0) {$TSX$};
\node(TX2) at (0,-1) {$TX$};
\node(TTX2) at (2,-1) {$TTX$};
\node(TX3) at (2,-2) {$TX$};

\node(diam) at (1.2,-0.7) {$\diamond$};
\draw[->] (STX) -- (TSX) node[pos=.5,above] {\tiny $\lambda$}; 
\draw[->] (TX2) -- (STX) node[pos=.5,left] {\tiny $e$}; 

\draw[->] (TX2) -- (TSX) node[pos=.5,left] {\tiny $Te$};

\draw[->] (TX2) -- (TTX2) node[pos=.5,below] {\tiny $\eta$};

\draw[->] (TSX) -- (TTX2) node[pos=.5,right] {\tiny $Ta$};

\draw[->] (TTX2) -- (TX3) node[pos=.5,right] {\tiny $\mu$};
\draw[->] (TX2) -- (TX3) node[pos=.5,below] {\tiny $\mathsf{id}$};
\end{tikzpicture}
\\
\begin{tikzpicture}
\node(STX) at (0,0) {$STX$};
\node(TSX) at (2,0) {$TSX$};
\node(TTX) at (4,0) {$TTX$};
\node(TX) at (6,0) {$TX$};
\draw[->] (STX) -- (TSX) node[pos=.5,above] {\tiny $\lambda$};  
\draw[->] (TSX) -- (TTX) node[pos=.5,above] {\tiny $Ta$};
\draw[->] (TTX) -- (TX) node[pos=.5,above] {\tiny $\mu$};
\node(TTTX2) at (4,-1) {$TTTX$};
\node(TTX2) at (6,-1) {$TTX$}; 
\draw[->] (TTTX2) -- (TTX2) node[pos=.5,above] {\tiny $\mu$};
\draw[->] (TTX2) -- (TX) node[pos=.5,right] {\tiny $\mu$};
\draw[->] (TTTX2) -- (TTX) node[pos=.5,left] {\tiny $T\mu$};
\node(diam) at (3,-1.5) {$\diamond$};
\node(TTSX3) at (4,-2) {$TTSX$};
\node(TSX3) at (6,-2) {$TSX$}; 
\draw[->] (TTSX3) -- (TSX3) node[pos=.5,below] {\tiny $\mu$};
\draw[->] (TTSX3) -- (TTTX2) node[pos=.5,left] {\tiny $T^2 a$};
\draw[->] (TSX3) -- (TTX2) node[pos=.5,right] {\tiny $Ta$};
\node(TSSX4) at (2,-3) {$TSSX$};
\node(TSTX4) at (4,-3) {$TSTX$};
\draw[->] (TSSX4) -- (TSTX4) node[pos=.5,below] {\tiny $TSa$};
\draw[->] (TSTX4) -- (TTSX3) node[pos=.5,right] {\tiny $T\lambda$};
\draw[->] (TSSX4) -- (TSX) node[pos=.5,left] {\tiny $Tm$};
\node(SSTX5) at (0,-4) {$SSTX$};
\node(STSX5) at (2,-4) {$STSX$};
\node(STTX5) at (4,-4) {$STTX$};
\node(STX5) at (6,-4) {$STX$};
\draw[->] (SSTX5) -- (STX) node[pos=.5,left] {\tiny $m$};
\draw[->] (SSTX5) -- (STSX5) node[pos=.5,below] {\tiny $S\lambda$};
\draw[->] (STSX5) -- (STTX5) node[pos=.5,below] {\tiny $STa$};
\draw[->] (STTX5) -- (STX5) node[pos=.5,below] {\tiny $S\mu$};
\draw[->] (STX5) -- (TSX3) node[pos=.5,right] {\tiny $\lambda$};
\draw[->] (STTX5) -- (TSTX4) node[pos=.5,right] {\tiny $\lambda$};
\draw[->] (STSX5) -- (TSSX4) node[pos=.5,right] {\tiny $\lambda$};
\end{tikzpicture}
\end{center}

\noindent The parts denoted by ($\diamond$) commute since $a:\overline{S}X\rightdcirc X = SX\to TX$ is in $\mathcal{EM}(\overline{S})$. The proof that $|h|=\mu_Y\circ Th$ is a morphism between $|a|$ and $|b|$ in $\mathcal{EM}(S)$ if $h$ is a morphism between $a$ and $b$ in $\mathcal{EM}(\overline{S})$ is straightforward and left to the reader.
\end{proof}

\subsection{Lawvere theories}\label{subsubsection:Lawvere_theories}

 Formally, a \emph{Lawvere theory}, or simply \emph{theory}, is a category whose objects are natural numbers $n\geq 0$ such that each $n$ is an $n$-fold coproduct of $1$.  For any element $i\in n$ let $i_n:1\to n$ denote the $i$-th coproduct injection and let the map $[ f_1, \ldots ,f_k ]:n_1+\ldots + n_k\to n$ be the cotuple of the family $\{f_l:n_l\to n\}_{l}$. Any morphism $k\to n$ of the form
$[i^1_n,\ldots,i^k_n]:k\to n$ for $i^j\in n$ is called \emph{base morphism} or \emph{base map}. Finally, let $!:n\to 1$ be defined by $!\defeq [1_1,1_1,\ldots,1_1]$.  A theory $\mathbb{T}$ is \emph{finitary} if $\mathbb{T}(m,n)$ is finite for any $n,m\geq 0$.

For any two theories $\mathbb{T}$ and $\mathbb{T}'$ a \emph{theory morphism} $h:\mathbb{T}\to \mathbb{T}'$ is an identity-on-objects functor. Let $\mathsf{Law}$ denote the category of theories as objects and theory morphisms as maps.

Any monad $S$ on $\Set$ induces a theory $\mathbb{T}_S$ associated with it by restricting the Kleisli category $\kl(S)$ to objects  $n$ for any $n\geq 0$.   Conversely, for any theory $\mathbb{T}$ there is a $\Set$ based monad $M_{\mathbb{T}}$ the theory is associated whose theory $\mathbb{T}_{M_\mathbb{T}}$ is isomorphic to $\mathbb{T}$ in $\mathsf{Law}$.  Hence, without any loss of generality we may assume $M_\mathbb{T}$ satisfies $M_\mathbb{T} (p)= \mathbb{T}(1,p)$ for any $p\geq 0$. Moreover, the assignments $S\mapsto \mathbb{T}_S$ and $\mathbb{T}\mapsto M_\mathbb{T}$ extend to functors $\mathsf{Mnd}\to \mathsf{Law}$ and $\mathsf{Law}\to \mathsf{Mnd}$\footnote{Here, $\mathsf{Mnd}$ denotes the category of all monads on $\Set$ as objects and monad morphisms as arrows. A natural transformation $h:T\implies T'$ is called a \emph{monad morphism} provided that it preserves unit and multiplication of the monad $T$, \ie $\eta' = h\circ \eta$ and $h\circ \mu = \mu'\circ hh$. }. Any theory morphism $h:\mathbb{T}\to \mathbb{T}'$ induces a monad morphism between the associated monads $M_\mathbb{T}$ and $M_\mathbb{T'}$ whose $p$-component  equals $h$, \ie it maps any $x\in M_\mathbb{T}(p)=\mathbb{T}(1,p)$ to $h(x)\in M_\mathbb{T'}(p)=\mathbb{T}'(1,p)$.  Conversely, any monad morphism $h:S\implies S'$ gives rise to a theory morphism which maps any $f:k\to Sl\in \mathbb{T}_S(k,l)$ to $  h_l\circ f:k\to S'l\in \mathbb{T}_{S'}(k,l)$ Since, in general, it is not the case that $M_{\mathbb{T}_S}$ is isomorphic to $S$ in $\mathsf{Mnd}$\footnote{ The monad $M_{\mathbb{T}_S}$ is isomorphic to $S$ whenever $S$ is \emph{finitary}, \ie when for any $x\in TX$ there exists a finite subset $X_0\subseteq X$ such that $x\in TX_0$.} this pair of functors does \emph{not} form an equilvalence of the categories $\mathsf{Law}$ and $\mathsf{Mnd}$. See \eg \cite{hyland:power:2007} for details.

\begin{example}\label{example:basic_theories}
Consider the theory associated with the LTS monad $\mathcal{P}(\Sigma^\ast\times \mathcal{I}d):\Set\to \Set$ from Example \ref{example:lts_monoid_lifting}. For any $i\in n$ the morphism $i_n:1\to \mathcal{P}(\Sigma^\ast\times n)$  in $\mathbb{T}_{\mathcal{P}(\Sigma^\ast\times \mathcal{I}d)}$ is given by $i_n(1) = \{(\varepsilon,i)\}$ and $!:n\to \mathcal{P}(\Sigma^\ast\times 1); i\mapsto \{(\varepsilon,1)\}$.  If $\mathbb{T}=\mathbb{T}_S$ is a theory associated with a $\Set$-based monad $(S,m,e)$ then $i_n:1\to n$ and $!:n\to 1$ in $\mathbb{T}_S$ are $\Set$-based maps $i_n:1\to Sn; 1\mapsto e_n(i)$ and $!:n\to S1; i\mapsto e_1(1)$. In general, all base morphisms in $\mathbb{T}_S$ are of the form $k\stackrel{f}{\to} n\stackrel{e_n}{\to}Sn$ for a $\Set$-map $f$.
\end{example}
\subparagraph{Models of Lawvere theory}\label{subsubsection:models_of_lawvere}  \emph{A model of a theory $\mathbb{T}$} is any product preserving functor $A:\mathbb{T}^{op}\to \Set$.  The category of models of $\mathbb{T}$ as objects and natural transformations between them as morphisms forms the category $\mathsf{Mod}\mathbb{T}$ of models of $\mathbb{T}$.  This category is known to be equivalent to $\mathcal{EM}(M_{\mathbb{T}})$ \cite{hyland:power:2007}.

\section{Omitted proofs}

\begin{proof}(Theorem \ref{lemma:preimage_homomorphism})
Assume $h$ is an algebra homomorphism from $m_X$ to $b$. Then we have:
\begin{center}
\resizebox{0.6\textwidth}{!}{
\begin{tikzpicture}[font=\scriptsize]
\node (xx) at(-1,0) {$SX$};
\node(sx) at (0,1) {$S^2X$};
\node (x) at (0,0) {$SX$};
\node(sy) at (2,1) {$S Y$};
\node (y) at (2,0) {$Y$};
\draw[->] (xx) -- (sx) node[pos=.2,above] {$Se$};
\draw[->] (xx) -- (x) node[pos=.5,below] {$\mathsf{id}$};
\draw[->] (x) -- (y) node[pos=.5,below] {$h$};
\draw[->] (sx) -- (x) node[pos=.5,right] {$m$};
\draw[->] (sy) -- (y) node[pos=.5,right] {$b$};
\draw[->] (sx) -- (sy) node[pos=.5,above] {$S h$};
\node (ineq) at (2.7,0.5) {${\implies}$};
\node (xx1) at(3,0) {$SX$};
\node(sx1) at (4,1) {$S^2X$};
\node (x1) at (4,0) {$SX$};
\node(sy1) at (6,1) {$S Y$};
\node (y1) at (6,0) {$Y$};
\draw[->] (sx1) -- (xx1) node[pos=.5,above] {$S e_-$};
\draw[->] (x1) -- (xx1) node[pos=.5,below] {$\mathsf{id}$};
\draw[->] (y1) -- (x1) node[pos=.5,below] {$h_-$};
\draw[->] (x1) -- (sx1) node[pos=.5,right] {$m_-$};
\draw[->] (y1) -- (sy1) node[pos=.5,right] {$b_-$};
\draw[->] (sy1) -- (sx1) node[pos=.5,above] {$S h_-$};
\end{tikzpicture}
}
\end{center}
A simple diagram chase gives us:
\begin{align*}
&h_-= S\left( (e_X)_-)\circ S ( h_-\right ) \circ b_-= S\left( (e_X)_-)\circ   h_-\right ) \circ b_-=S\left( (h\circ e_X)_-\right ) \circ b_-.
\end{align*}
Conversely, take $f:Y\to X$ and consider $h:SX\to Y=SX\stackrel{S f_-}{\to }SY \stackrel{b}{\to} Y$.  It is easy to prove that $h$ is an algebra homomorphism from $m_X$ to $b$. Moreover, $h_- = Sf \circ b_-$.  This completes the proof.
\end{proof}

\begin{proof}(Proposition \ref {proposition:duality_free_monads})
Consider any EM-algebra $a:F^\ast X\to X$. Then the map $\underline{a} = a\circ \nu:FX\to X$ satisfies $\nu \circ\underline{a}_- = \nu \circ \nu_- \circ a_- \leq a_-$ and, hence, we have:
$$
(\nu \circ a_-)^\ast \leq (a_-)^\ast = a_-.
$$
By our assumptions the map $(\nu\circ \underline{a}_-)^\ast_-:F^\ast X\to X$ is an Eilenberg-Moore algebra for the monad $F^\ast$. Since by our assumptions $a:F^\ast X\to X$ was the least EM-algebra greater than $\underline{a}\circ \nu_-= a\circ \nu\circ \nu_-$ we have $(\nu \circ a_-)^\ast _- = a_-$. This completes the proof.
\end{proof}

\begin{proof}(Theorem \ref{theorem:regular_maps_form_subtheory})
Let us first prove the following lemma:

\begin{lemma}\label{lemma:cotupling_regular}
Let $r_1,\ldots,r_p:1\rightdcirc \overline{S}q= 1\to TSq$ be regular morphisms. Then the map $[r_1,\ldots,r_p] \colon p\rightdcirc \overline{S}q$ is of the form  $[r_1,\ldots,r_p] = p\stackrel{\psi}{\rightdcirc} m \stackrel{\alpha}{\rightdcirc} \overline{S}m\stackrel{\overline{S}\chi}{\rightdcirc} \overline{S}q$ for  $\alpha\in \mathcal{SAT}(\overline{S})$ and $\kl(T)$ map $\chi$ and a base morphism $\psi$ for $T$.
\end{lemma}

\begin{proof}(Lemma \ref{lemma:cotupling_regular})
We take $r_i = 1\stackrel{\psi_i}{\rightdcirc } n_i \stackrel{\alpha_i} {\rightdcirc } \overline{S}n_i \stackrel{\overline{S} \chi_i}{\rightdcirc }\overline{S}q$, where $\alpha_i\in \mathcal{SAT}(\overline{S})$ for $i=1,\ldots, p$ and $\psi_i$ is a base morphism for $T$ and we show that there is $\alpha\in \mathcal{SAT}(\overline{S})$ such that $[r_1,\ldots,r_p] \stackrel{\diamond}{=} p\stackrel{\psi}{\rightdcirc} m \stackrel{\alpha}{\rightdcirc} \overline{S}m\stackrel{\overline{S}\chi}{\rightdcirc} \overline{S}q$. Define $m= n_1+\ldots+ n_p$ and put $\alpha:m\to \overline{S}m$ to be given by 
$$
m=n_1+\ldots +n_p \stackrel{\alpha_1+\ldots + \alpha_p}{\rightdcirc}\overline{S}n_1+\ldots + \overline{S}n_p \stackrel{[\overline{S}\mathsf{in}^{n_i}_{m}]_i}{\rightdcirc}\overline{S}(n_1+\ldots + n_p) = \overline{S}m
$$
By assumption (\ref{assumption:reg2.5}) $\alpha$ is a member of $\mathcal{SAT}(\overline{S})$. Moreover, take $\psi = [\psi_1,\ldots,\psi_p]$ and $\chi = [\chi_1,\ldots,\chi_p]$. Then $\chi$ is in $\kl(T)$, $\psi$ is a base morphism and the equality ($\diamond$) holds. 
\end{proof}
Note that by (\ref{assumption:reg2}) the identity maps in $\mathbb{T}_{TS}$, i.e. $\overline{e}_n:n\rightdcirc\overline{S}n=n\to TSn$ are regular. Moreover, by Lemma \ref{lemma:cotupling_regular} regular maps are closed under coptupling. Finally, if two maps $r_1:n\rightdcirc \overline{S}k$ and $r_2:k\rightdcirc\overline{S} q$ are regular then their composition in $\kl(\overline{S})=\kl(TS)$ is given by
$$
n\stackrel{r_1}{\rightdcirc} \overline{S}k\stackrel{\overline{S}r_2}{\rightdcirc }\overline{S}^2 q \stackrel{\overline{m}}{\rightdcirc}\overline{S}q 
$$
and is regular by (\ref{assumption:reg3}). This completes the proof of Theorem \ref{theorem:regular_maps_form_subtheory}.
\end{proof}

\begin{proof}(Theorem \ref{lemma:from_regular_to_morphisms}) In the first part of the proof we show that if a map $L:1\rightdcirc \overline{S}p=1\to \mathcal{P}Sp\in \mathbb{T}_{\mathcal{P}S}(1,p)$ is regular then set $\{f:1\to Sp\in \mathbb{T}_S(1,p) \mid f(1) \subseteq L(1)\}$ is recognizable. We present  a construction of a theory morphism $\mathbb{T}_S\to \mathbb{T}'$ which recognizes a given regular map (\ref{exp:regular}). We have:
$$
\infer{|\mathsf{Alg}(\alpha)|:S\mathcal{P}n\to \mathcal{P}n \in \mathcal{EM}(S), }{\infer{\mathsf{Alg}(\alpha):Sn\to \mathcal{P}n }{\infer{\mathsf{Alg}(\alpha)\defeq \alpha_-:\overline{S}n\rightdcirc n \quad  \in \mathcal{EM}(\overline{S})}{{\alpha:n\rightdcirc \overline{S}n \quad \in \mathcal{SAT}(\overline{S})}}}}
$$
where $|\mathsf{Alg}(\alpha)| \defeq S\mathcal{P}n \stackrel{\lambda}{\to}\mathcal{P}Sn \stackrel{\mathcal{P}\mathsf{Alg}(\alpha)}{\to} \mathcal{P}\mathcal{P}n \stackrel{\bigcup}{\to}\mathcal{P}n $. By Theorem~\ref{theorem:EM_on_kleisli_to_EM_on_base} the algebra $|\mathsf{Alg}(\alpha)|$ is,  indeed, an Eilenberg-Moore algebra for the monad $S$ on $\Set$.  By Subsection \ref{subsubsection:models_of_lawvere} any algebra in $\mathcal{EM}(S)$ induces a model in $\mathsf{Mod}\mathbb{T}_S$. So, we continue:
$$
\infer{A:\mathbb{T}_S^{op} \to \Set\quad  \in \mathsf{Mod}\mathbb{T}_S}{|\mathsf{Alg}(\alpha)|:S\mathcal{P}n\to \mathcal{P}n \quad  \in \mathcal{EM}(S) }
$$
The explicit recipe for the model $A$ is as follows.
\begin{center}
\resizebox{0.8\textwidth}{!}{
\begin{tikzpicture}
\node (ak) at (0.5,0) {$k\to \mathcal{P}n$};
\node (k) at (0.5,1) {$k$};
\draw[|->] (k) -- (ak) node[pos=.5,left] {\tiny $A$};

\node (af) at (3.5,0) {$f_A:(l\to \mathcal{P}n)\to (k\to \mathcal{P}n)$};
\node (f) at (3.5,1) {$f:k\to Sl \in \mathbb{T}_S(k,l)$};
\draw[|->] (f) -- (af) node[pos=.5,left] {\tiny $A$};
\node(where) at (6,0.5) { \small where };

\node (kpn) at (8,0.1) {$k\stackrel{f}{\to} Sl \stackrel{Sx}{\to} S\mathcal{P}n \stackrel{|\mathsf{Alg}(\alpha)|}{\to} \mathcal{P}n$};
\node (lpn) at (8,1) {$x:l\to \mathcal{P}n$};
\draw[|->] (lpn) -- (kpn) node[pos=.5,left] {\tiny $f_A$};
\end{tikzpicture}
}
\end{center}
Now consider the variety $\mathsf{V}(A)$ generated by $A$ which consists of all models of $\mathbb{T}_S$ obtained from $A$ by applying three types of operators: $\mathsf{H}, \mathsf{S}$ and $\mathsf{P}$. We know that $\mathsf{V}(A) = \mathsf{HSP}(A)$. The category $\mathsf{V}(A)$ admits all free objects $F_A(X):\mathbb{T}_S^{op}\to \Set$ (\eg \cite{bs81}) with their explicit  description left as an exercise for the reader in \eg \cite[Exercise 11.5]{bs81}. Whenever $X=m$ these algebras are described explicitly by:
\vspace{-0.3cm}
\begin{center}
\resizebox{0.7\textwidth}{!}{
\begin{tikzpicture}
\node (ak) at (0,0) {$\{g_A:(m\to \mathcal{P}n)\to (k\to \mathcal{P}n)\mid g:k\to Sm\in \mathbb{T}_S(k,m)\}$};
\node (k) at (0,1) {$k$};
\draw[|->] (k) -- (ak) node[pos=.5,left] {\tiny $F_{A}(m)$};

\node (kpn) at (0,-2) {$F_{A}(m)(l)\to F_{A}(m)(k); g_A\mapsto (g\diamond f)_A\stackrel{\dagger}{=} f_A\circ g_A$};
\node (lpn) at (0,-1) {$f:k\to Sl \in \mathbb{T}_S(k,l)$};
\draw[|->] (lpn) -- (kpn) node[pos=.5,left] {\tiny $F_{A}(m)$};
\end{tikzpicture}
}
\end{center}
where in the above, $\diamond$ denotes the composition in $\mathbb{T}_S$ and the identity marked with $\dagger$ follows by a simple diagram chase and is proven in Theorem \ref{theorem:theory_from_free_algebra}. The free algebra functor $F_A(-):\Set\to \mathsf{V}(A)$ induces a $\Set$-monad which induces a theory. This theory is explicitly described in the following statement.

\begin{theorem}\label{theorem:theory_from_free_algebra}
Let $\mathbb{T}_{F_A}$ be the category whose objects are $n$ for $n\geq 0$ and morphisms from $k$ to $l$ are $\mathbb{T}_{F_A}(k,l)=F_A(l)(k)$ with the composition of the morphisms $(h_1)_A\in \mathbb{T}_{F_A}(k,l)$ and  $(h_2)_A\in \mathbb{T}_{F_A}(l,m)$ given by $$(h_2\diamond h_1)_A \stackrel{\dagger}{=} (h_1)_A\circ (h_2)_A\in \mathbb{T}_{F_A}(k,m).$$ Then $\mathbb{T}_{F_A}$ is a well defined Lawvere theory.
\end{theorem}
\begin{proof}
At first we will prove that $$(h_2\diamond h_1)_A = (h_1)_A\circ (h_2)_A.$$ For  $x:m\to \mathcal{P}n$ the image $(h_2\diamond h_1)_A(x)$ equals $$(h_2\diamond h_1)_A(x) = k\stackrel{h_1}{\to} Sl\stackrel{Sh_2}{\to} S^2 m \stackrel{m}{\to}Sm\stackrel{Sx}{\to} S\mathcal{P}n\stackrel{|\mathsf{Alg}(\alpha)|}{\to} \mathcal{P}n.$$
The desired equality follows by commutativity of the diagram below.
\begin{center}
\resizebox{0.7\textwidth}{!}{
\begin{tikzpicture}
\node(k) at (0,-2) {$k$};
\node(Sl) at (2,-2) {$Sl$};
\node(SSm) at (4,0) {$S^2 m$};
\node(Sm) at (6,0) {$S m$};

\node(SPn) at (6,-3) {$S\mathcal{P}n$};
\node(Pn) at (6,-4) {$\mathcal{P}n$};
\node(SSPn) at (4,-3) {$S^2 \mathcal{P}n$};
\node(SPn2) at (4,-4) {$S\mathcal{P}n$};

\draw[->] (k) -- (Sl) node[pos=.5,above] {\tiny $h_1$};
\draw[->] (Sl) -- (SPn2) node[pos=.5,left] {\tiny $S(|\mathsf{Alg}|\circ Sx\circ h_2)$};
\draw[->] (Sl) -- (SSm) node[pos=.5,above] {\tiny $Sh_2$};

\draw[->] (SSm) -- (Sm) node[pos=.5,above] {\tiny $m$};
\draw[->] (SSPn) -- (SPn) node[pos=.5,above] {\tiny $m$};
\draw[->] (Sm) -- (SPn) node[pos=.5,right] {\tiny $Sx$};
\draw[->] (SSm) -- (SSPn) node[pos=.5,right] {\tiny $S^2x$};
\draw[->] (SPn) -- (Pn) node[pos=.5,right] {\tiny $|\mathsf{Alg}(\alpha)|$};
\draw[->] (SPn2) -- (Pn) node[pos=.5,above] {\tiny $|\mathsf{Alg}(\alpha)|$};
\draw[->] (SSPn) -- (SPn2) node[pos=.5] {\tiny $S|\mathsf{Alg}(\alpha)|$};
\end{tikzpicture}
}
\end{center}
This proves that $\mathbb{T}_{F_A}$ is a well-defined category. A simple verification leads to a conclusion that it admits coproducts and that $n=1+\ldots+1$ for any object $n$. Hence, $\mathbb{T}_{F_A}$ is a theory.
\end{proof}

Now let us get back to the proof of Theorem \ref{lemma:from_regular_to_morphisms}. Note that $\mathbb{T}_{F_A}$is a finitary theory which comes with a theory morphism $h:\mathbb{T}_S\to \mathbb{T}_{F_A}$ mapping $n$ to $n$ and  \begin{align}
f:k\to Sl \in \mathbb{T}_S(k,l) \text{ to }h(f)=f_A\in \mathbb{T}_{F_A}(k,l).\tag{TM} \label{equation:theory_morphism}
\end{align}

Recall that $\chi:n\to \mathcal{P}p$, consider $\chi_-:p\to \mathcal{P}n$ and take the set \begin{align*}
&T_\chi\subseteq \mathbb{T}_{F_A}(1,p)= \{g_A:(p\to \mathcal{P}n)\to (1\to \mathcal{P}n)\mid g:1\to Sp\in \mathbb{T}_S(1,p)\}
\end{align*}
defined by
$
T_\chi = \{g_A\mid g_A(\chi_-)=i_n \},
$
where $i_n:1\to \mathcal{P}n$ maps the unique element of  $1$ onto  $\{i\}$ and is the lifting of  the $\Set$-map $1\to n; 1\mapsto i$ to $\kl(\mathcal{P})$. For $f:1\to Sp\in \mathbb{T}_S(1,p)$  we then have the following chain of equivalences:
$f\in h^{-1}(T_\chi)$ iff $f_A(\chi_-)=i_n$ iff
\begin{center}
\resizebox{0.7\textwidth}{!}{
$$
\infer{f\in \{f':1\to Sp\mid f'(1) \subseteq \overline{S}\chi\bullet \alpha\bullet i_n (1) \} }{\infer{\overline{S}\chi\bullet \alpha \bullet i_n  \geq f^\sharp}{
\infer{(\overline{S}\chi\bullet \alpha \bullet i_n  )_- \geq (\overline{S}\chi\bullet \alpha \bullet i_n)_- \bullet f^\sharp\bullet f^\sharp_- = (i_n)_- \bullet i_n \bullet f^\sharp_-\geq f^\sharp_-}{
\infer{(i_n)_- \bullet (\overline{S}\chi\bullet \alpha )_- \bullet f^\sharp\bullet f^\sharp_- = (i_n)_- \bullet i_n \bullet f^\sharp_-}{
\infer{(\overline{S}\chi\bullet \alpha )_- \bullet f^\sharp\bullet f^\sharp_- = i_n \bullet f^\sharp_-}{\infer{(\overline{S}\chi\bullet \alpha)_- \bullet f^\sharp = i_n }{\infer{\alpha_- \bullet \overline{S}\chi_- \bullet f^\sharp = i_n }{\infer{\mathsf{Alg}(\alpha)\bullet \overline{S}\chi_-\bullet f^\sharp = i_n , \quad\in \kl(\mathcal{P})}{\infer{\bigcup \circ \mathcal{P}\mathsf{Alg}(\alpha)\circ \lambda\circ S\chi_-\circ f=i_n}{|\mathsf{Alg}(\alpha)|\circ S\chi_-\circ f = i_n \quad \in \Set}}}}}}}}}
$$
}
\end{center}
\noindent In the above $\bullet$ denotes the composition in $\kl(\mathcal{P})$. This completes the first part of the proof of Theorem \ref{lemma:from_regular_to_morphisms}.

Now, in order to see the converse assume there is a finitary theory $\mathbb{T}'$ and a theory morphism $h:\mathbb{T}_S\to \mathbb{T}'$. Consider a set $T\subseteq \mathbb{T}'(1,p)$. Our aim is to represent $h^{-1}(T)$ in terms of an expression of the form (\ref{exp:regular}).
 As mentioned in Subsection~\ref{subsubsection:Lawvere_theories}, without loss of generality, we may assume that
 $M_{\mathbb{T}_S} p=\mathbb{T}_S(1,p)$ and $M_\mathbb{T'} p = \mathbb{T}'(1,p)$. Since $S$ was taken to be finitary, the monad $M_{\mathbb{T}_S}$ is isomorphic to $S$ in $\mathsf{Mnd}$. Hence, in what follows we will slightly abuse the notation and denote the multiplication and unit of $M_{\mathbb{T}_S}$ by $m$ and $e$ respectively.

  The theory morphism $h:
 \mathbb{T}_S\to \mathbb{T}'$ induces a monad morphism, which will be denoted by
 $$\overline{h}:M_{\mathbb{T}_S}\implies M_{\mathbb{T}'}.$$ Following Subsection
 \ref{subsubsection:Lawvere_theories}, the $p$-component of the morphism $\overline{h}$ is
 $h_p:M_{\mathbb{T}_S}p=\mathbb{T}_S(1,p) \to M_{\mathbb{T}'}p = \mathbb{T}'(1,p); f\mapsto h(f)$. Consider the
 $p$-component of the multiplication $m'$ of $M_\mathbb{T'}$ and note that the $M_{\mathbb{T}_S}$-algebra
\begin{align*}
a:M_{\mathbb{T}_S}M_{\mathbb{T}'}(p)\stackrel{\overline{h}_{M_{\mathbb{T}'}p}}{\to} M_{\mathbb{T}'}M_{\mathbb{T}'}(p)\stackrel{m'_p}{\to} M_{\mathbb{T}'}p=\mathbb{T}'(1,p).
\end{align*}
is, in fact, a member of $\mathcal{EM}(M_{\mathbb{T}})$. This follows by the properties of the monad morphism $\overline{h}$ and a simple diagram chase. By the same properties we have the following:
\begin{center}
\resizebox{0.7\textwidth}{!}{
\begin{tikzpicture}
\node(sx) at (-4,2) {$M_{\mathbb{T}_S}M_{\mathbb{T}_S}p=M_{\mathbb{T}_S}\mathbb{T}_S(1,p)$};
\node (x) at (-4,0) {$M_{\mathbb{T}_S}p=\mathbb{T}_S(1,p)$};
\node(st') at (2,2) {$M_{\mathbb{T}_S}M_{\mathbb{T}'}p=M_{\mathbb{T}_S}\mathbb{T}'(1,p)$};
\node(rt') at (2,1) {$M_{\mathbb{T}'}M_{\mathbb{T}'}p=M_{\mathbb{T}'}\mathbb{T}'(1,p)$};
\node (t') at (2,0) {$M_{\mathbb{T}'}p=\mathbb{T}'(1,p)$};
\draw[->] (sx)--(st') node[pos=.5,above] {\tiny $M_{\mathbb{T}_S}h$};
\draw[->] (x) -- (t') node[pos=.5,below] {\tiny $h$};
\draw[->] (sx) -- (x) node[pos=.5,left] {\tiny $m$};
\draw[->] (st') -- (rt') node[pos=.5,right] {\tiny $\overline{h}$};
\draw[->] (rt') -- (t') node[pos=.5,right] {\tiny $m'$};
\end{tikzpicture}
}
\end{center}
Consider the lifting $a^\sharp\defeq \{-\}\circ a: M_{\mathbb{T}_S}\mathbb{T}'(1,p)\to \mathcal{P}\mathbb{T}'(1,p)$ of $a$ to $\kl(\mathcal{P})$. The map $a^\sharp$ is an $\overline{M_{\mathbb{T}_S}}$-algebra which is a member of $\mathcal{EM}(\overline{M_{\mathbb{T}_S}})$.  By the diagram above we have:
$$
a^\sharp \bullet \overline{M_{\mathbb{T}_S}}h^\sharp = h^\sharp\bullet m^\sharp.
$$
Theorem \ref{lemma:preimage_homomorphism} implies:
$h^\sharp_- = \overline{M_\mathbb{T}}(e^\sharp_- \bullet h^\sharp_-)\bullet \mathsf{CoAlg}(a^\sharp)$.
Hence, for any $t\in T\subseteq \mathbb{T}'(1,p)$ we have
\begin{align*}
&h^{-1}(t) = h^\sharp_-(t) = h^\sharp_-\bullet t_{\mathbb{T}'}^\sharp(1) =\overline{M_{\mathbb{T}_S}}(e^\sharp_- \bullet h_-)\bullet \mathsf{CoAlg}(a^\sharp)\bullet t_{\mathbb{T}'}^\sharp(1),
\end{align*}
where $t_{\mathbb{T}'}:1\to \mathbb{T}'(1,p)$ maps the unique element of $1$ to $t\in \mathbb{T}'(1,p)$.
This exactly shows that the map $1\to \mathcal{P}Sp$ which assigns to the unique element of $1$ the set $h^{-1}(t)\subseteq 1\to Sp$ is regular. Now, since $T$ is finite we have $h^{-1}(T) = h^{-1}(t_1)\cup\ldots h^{-1}(t_n)$ for $T=\{t_1,\ldots,t_n\}$. Hence, by the assumption (\ref{assumption:reg2}) and (\ref{assumption:reg3}) we get that maps of the form $1\stackrel{\psi}{\rightdcirc }n\stackrel{\alpha}{\rightdcirc}\overline{S}n\stackrel{\overline{S}\chi}{\rightdcirc }\overline{S}p$ are regular for any $\psi,\chi$ in $\kl(\mathcal{P})$ and $\alpha\in \mathcal{SAT}(\overline{S})$. This precisely means that regular maps are closed under finite unions. Therefore, we get the desired conclusion. This ends the proof of Theorem \ref{theorem:main}.

\end{proof}

\section{Examples}
\label{section:appendix_examples}

In this appendix we illustrate the generality of the results presented by listing some representative examples of models fitting our framework besides our running example of non-deterministic automata and regular languages.

\subsection{Tree automata and their languages}
\label{section:appendix_examples_tree}

Here, we focus on \emph{non-deterministic tree automaton}, \ie a tuple $(Q,\Sigma,\delta,\mathcal{F})$, where $\delta:Q\times \Sigma \to \mathcal{P}(Q\times Q)$ and the rest is as in the case of standard non-deterministic automata (\eg \cite{pin:automata}).

\subparagraph{Trees}
Formally, a \emph{binary tree} or simply \emph{tree} with nodes in $A$ is a function $t: P \to A$, where $P$ is a non-empty prefix closed subset of $\{l,r\}^\ast$. The set $P\subseteq \{l,r\}^\ast$ is called the \emph{domain} of $t$ and is denoted by $\mathsf{dom}(t)\defeq P$. Elements of $P$ are called \emph{nodes}. For a node $w\in P$ any node of the form $wx$ for $x\in \{l,r\}$ is called a \emph{child} of $w$. A tree is called \emph{complete} if all nodes have either two children or no children.  A \emph{height} of a tree $t$ is $\max \{ |w| \mid w\in \mathsf{dom}(t)\}$. A tree $t$ is \emph{finite} if it is of a finite height. The \emph{frontier} of a tree $t$ is $\mathsf{fr}(t) \defeq \{ x\in \mathsf{dom}(t)\mid x \{l,r\} \cap P = \varnothing \}$. Elements of $\mathsf{fr}(t)$ are called \emph{leaves}. Nodes from $\mathsf{dom}(t)\setminus \mathsf{fr}(t)$ are called \emph{inner nodes}. The \emph{outer frontier} of $t$ is defined by $\mathsf{fr}^+(t) \defeq \mathsf{dom}(t)\{l,r\}\setminus \mathsf{dom}(t)$. \ie it consists of all the words $wi\notin \mathsf{dom}(t)$ such that $w\in \mathsf{dom}(t)$ and $i\in \{l,r\}$. Finally, set $\mathsf{dom}^+(t) \defeq \mathsf{dom}^+(t) \cup \mathsf{fr}^+(t)$.

Let $T_{\Sigma}X$ denote the set of all finite complete trees $t:P\to \Sigma+X$ with inner nodes taking values in $\Sigma$ and which have leaves from the set $X$. Note that trees from $T_{\Sigma}X$ of height $0$ can be thought of as elements of $X$. Hence, we may write $X\subseteq T_\Sigma X$.

\subparagraph{Languages}
Let $\mathcal{Q}=(Q,\Sigma,\delta,\mathcal{F})$ be a tree automaton. A \emph{run} of the automaton $\mathcal{Q}$ on a finite tree $t\in T_\Sigma(1)$ starting at the state $s\in Q$ is a map $\mathfrak{r}:\mathsf{dom}^+(t)\to Q$ such that $\mathfrak{r}(\varepsilon)=s$ and for any $x\in \mathsf{dom}(t)\setminus \mathsf{fr}(t)$ we have
$
(\mathfrak{r}(xl),\mathfrak{r}(xr))\in \delta(\mathfrak{r}(x),t(x)).
$
We say that the run $\mathfrak{r}$ is \emph{successful} if $\mathfrak{r}(w)\in \mathcal{F}$ for any $w\in \mathsf{fr}^+(t)$ for the tree $t$. The set of finite trees recognized by a state $s$ in $\mathcal{Q}$ is defined as the set of finite trees  $t\in T_\Sigma (1)$ for which there is a run  in $\mathcal{Q}$ starting at $s$ which is succesful on the tree $t$.

\subparagraph{Tree automata and regular behaviours}

Let  $\mathcal{Q}=(Q,\Sigma,\delta,\mathcal{F})$ be a finite non-deterministic tree automaton. Since $Q$ is finite without any loss of generality we may assume $Q=n$ for some $n\geq 0$. Moreover, in this case, $\delta:Q\times \Sigma\to \mathcal{P}(Q\times Q)$ can be given in terms of
$$
\alpha:n\to \mathcal{P}(n\times \Sigma \times n); i\mapsto \{(j,a,k)\mid (j,k)\in \delta(a,i)\}.
$$
Hence, any non-deterministic tree automaton can be viewed as a pair $(\alpha:n\to \mathcal{P}(n\times \Sigma\times n),\mathcal{F})$. Since $n\times \Sigma \times n$ is, in fact, the set of binary trees from $T_\Sigma$ of height one, the pair becomes $(\alpha:n\to \mathcal{P}T_\Sigma n, \mathcal{F})$. By \cite{brengos2018:concur}, the
map $L(\alpha,\mathcal{F},i):1\to \mathcal{P}T_\Sigma 1 = 1\stackrel{i_n}{\rightdcirc} n \stackrel{\alpha^\ast}{\rightdcirc} \overline{T_\Sigma} n\stackrel{\overline{T_\Sigma}\chi_{\mathcal{F}}}{\rightdcirc} \overline{T_\Sigma} 1$, where $\chi_\mathcal{F}$ is as in Section \ref{section:classical_regular_revisited}, satisfies:
\begin{align*}
&t\in L(\alpha,\mathcal{F},i)(1) \iff t \text{ is accepted by the state $i$} \text{ in the tree automaton }(\alpha,\mathcal{F}).
\end{align*}

The rest of this section will be devoted to considering an example of a regular tree language which we will characterize in terms of Lawvere theory morphism recognition. This characterisation will be given in more details (compared to the sketch presented in Section \ref{section:examples}). We assume the reader is familiar with the proof of Theorem \ref{lemma:from_regular_to_morphisms}. Let $\Sigma = \{a,b\}$ and consider a two state automaton whose transition map $\alpha:2\to \mathcal{P}(2\times \Sigma \times 2)$ is defined by $1\mapsto \varnothing,  2\mapsto \{(2,a,2),(1,b,1)\}$
and $\mathcal{F}=\{1\}$. It is easy to see that $L(\alpha,\mathcal{F},2)(1)\in \mathcal{P}T_\Sigma 1$ consists of binary trees of height $>0$ whose nodes preceding the leaf nodes are all $b$ and the remaining non-leaf nodes are $a$. Indeed,  the saturated map $\alpha^\ast:2\to \mathcal{P}(T_\Sigma 2)$ satisfies $\alpha^\ast(1) = \{1\}$ and $\alpha^\ast(2)$ contains the trees: $2$, $(1,b,1)$ and satisfies the following implication: if $t,t'\in \alpha^\ast(2)$ then $(t,a,t')$ is in $\alpha^\ast(2)$. The morphism $L(\alpha,\mathcal{F},2):1\rightdcirc \overline{T}_\Sigma 1=1\to \mathcal{P}T_\Sigma 1$ is exactly as stated above since
$L(\alpha,\mathcal{F},2) = 1\stackrel{2_2}{\rightdcirc} 2\stackrel{\alpha^\ast}{\rightdcirc }\overline{T_\Sigma}2 \stackrel{\overline{T_\Sigma}\chi_{\{1\}}}{\rightdcirc} \overline{T_\Sigma}1.$

The map $\mathsf{Alg}(\alpha^\ast)=\alpha^\ast_-:T_\Sigma 2\to \mathcal{P}2$ assigns $1\mapsto \{1\}$, $2\mapsto \{2\}$ and for any $t\in \alpha^\ast(2)$, $t\mapsto \{2\}$. The remaining trees from $T_\Sigma 2$ are assigned to $\varnothing$.  This yields a $4$-element algebra $|\mathsf{Alg}(\alpha^\ast)|:T_\Sigma \mathcal{P}2 \to \mathcal{P} 2$ which takes any tree in $T_\Sigma \mathcal{P}2$ with a leaf equal $\varnothing$  to $\varnothing$ and any other tree $T\in T_\Sigma \mathcal{P}2$ is mapped onto
$
T\mapsto \bigcup \{ \mathsf{Alg}(\alpha^\ast)(t)=\alpha^\ast_-(t)\mid t \ll T \},
$
where $t\ll T$ means that $t$ is obtained from $T$ by replacing all subset leaves of $T$ by one of their elements. In particular,  $(x,a,\{1\}),(\{1\},a,x),(x,\sigma,\varnothing),(\varnothing,\sigma,x),(\{2\},b,\{2\}) \mapsto \varnothing$ for any $x,x'\in \mathcal{P}2$, $(x,a,x')\mapsto \{2\}$ if $2\in x,x'$ and $(x,b,x')\mapsto \{1\}$ if $1\in x,x'$. Since, $|\mathsf{Alg}(\alpha^\ast)|$ is an Eilenberg-Moore algebra for $T_\Sigma$, all other values are uniquely determined by these.

Now, the carrier of the free algebra $F_A(1)$ over $1$, \ie $F_A(1)(1)$, is a submonoid of the monoid of maps $(1\to \mathcal{P}2)\to (1\to\mathcal{P}2)$ (or equivalently, of maps $\mathcal{P}2\to \mathcal{P}2$)  generated by the assignments:
$a_A:\mathcal{P}2\to \mathcal{P}2$ and $b_A:\mathcal{P}2\to \mathcal{P}2$  defined by $a_A(x) = (x,a,x)$ and $b_A(x)=(x,b,x)$. It turns out that it has $4$ elements which are given in the following table.
\begin{wrapfigure}[5]{l}{0.5\textwidth}
\vspace{-0.4cm}
\resizebox{0.5\textwidth}{!}{
$\begin{array}{c|c|c|c|c|}
 & \mathsf{id} & a_A=a_A^2 & b_A=a_A\circ b_A  & b_A^2=b_A\circ a_A\\\hline
\varnothing &\varnothing &\varnothing  &\varnothing &\varnothing\\ \hline
\{1\} & \{1\} & \varnothing & \{2\} &  \varnothing  \\ \hline
\{2\} & \{2\} & \{2\} & \varnothing  & \varnothing  \\ \hline
\{1,2\} & \{1,2\} & \{2\} & \{2\} &\varnothing  \\ \hline
\end{array}$
}
\end{wrapfigure}
Now, in our case, $T_\chi$ consists of maps from the above table that assign to $\{1\}$ the value $\{2\}$, \ie $T_\chi=\{b_A\}$, and the induced Lawvere theory morphism $\mathbb{T}_{T_\Sigma}\to \mathbb{T}_{F_A}$ restricted to $\mathbb{T}_{T_\Sigma}(1,1)$ maps the tree $1$ to $\mathsf{id}$, any tree $t$ such that after the composition with the tree $(1,b,1)$ in $\mathbb{T}_{T_\Sigma}$ the result is in $L(\alpha,\mathcal{F},2)$ to $a_A$, any tree from $L(\alpha,\mathcal{F},2)$ onto $b_A$ and any other to $b^2_A$. We see that it is a monoid homomorphism\footnote{Note that the composition in $\mathbb{T}_{F_A}(1,1)$ is defined to be the composition of $F_A(1)(1)$ in the reversed order. (\emph{conf.} Theorem \ref{theorem:theory_from_free_algebra}).} and that $L(\alpha,\mathcal{F},2)=h^{-1}(\{b_A\})$.

\subsection{Weighted automata and their languages}
\label{section:appendix_examples_weighted}

\subparagraph{Generalised countable multisets}

A semiring $(\mathcal{S},+,0,\cdot,1)$ is said to be 
\emph{positively ordered} whenever it can be equipped with a partial
order $(\mathcal{S},\leq)$ such that the unit $0$ is the bottom element of this ordering and semiring operations are monotonic in both components \ie $x \leq y$  implies $x \mathrel{\diamond} z \leq y \mathrel{\diamond} z$ and 
$z \mathrel{\diamond} x \leq z \mathrel{\diamond} y$ for $\diamond \in 
\{+,\cdot\}$ and $x,y,z\in \mathcal{S}$.  
A semiring is positively ordered if and only if it is \emph{zerosumfree} \ie $x + y = 0$ implies $x = y = 0$; the natural order $x \lhd y
\iff \exists z. x + z = y$ is the weakest one rendering $\mathcal{S}$ positively
ordered.

A positively ordered semiring is said to be \emph{$\omega$-complete} if it has
countable sums given as
$\sum_{i \leq \omega} x_i = \sup\{
		\sum_{j\in J} x_j \mid
		J \subset \omega \}
$. It is called \emph{$\omega$-continuous} if suprema of ascending
$\omega$-chains exist and are preserved by both operations
\ie:
$y \mathrel{\diamond} \bigvee_{i} x_i = \bigvee_{i} y \mathrel{\diamond} x_i$
and
$\bigvee_{i} x_i \mathrel{\diamond} z	= \bigvee_{i} x_i \mathrel{\diamond} z$
for $\diamond \in \{+,\cdot\}$ and $x,y,z\in \mathcal{S}$.
Examples of such semirings are:
the boolean semiring, the arithmetic semiring of non-negative real numbers with infinity and the tropical semiring. 

Henceforth we assume $(\mathcal{S},+,0,\cdot,1,\leq)$ to be a positive, $\omega$-complete, $\omega$-continuous semiring.

A countable $\mathcal{S}$-multiset is a pair $(X,\phi)$ where $X$ is a set and $\phi\colon X \to \mathcal{S}$ is a function such that the set $\supp(\phi) = \{x \mid \phi(x) > 0 \}$ (called \emph{support}) is countable. We will abuse the notation and simply write $\phi$ for the multiset $(X,\phi)$. We often write $\mathcal{S}$-multiset as formal sums \ie sums of expressions of form $s \bullet x$ where $x$ is an element and $s$ its membership degree.

The $\mathcal{S}$-multiset functor $\mathcal{M_S}\colon \Set \to \Set$ assigns to every set $X$ the set $\mathcal{M_S}X$ of $\mathcal{S}$-multisets with universe $X$, and to every function $f\colon X \to Y$ the function mapping each $(X,\phi)$ to $(X, \sum_{x \in \supp(\phi)} \phi(x) \bullet f(x))$. 
This functor carries a monad structure $(\mathcal{M_S},\mu,\eta)$ whose multiplication $\mu$ and unit $\eta$ are given on each set $X$ by the mappings:
\[(\mathcal{M_S}X,\psi) \mapsto \left(X,\sum (\phi(x) \cdot \psi(\phi)) \bullet x\right)
\qquad x \mapsto (X,x \bullet 1)
\text{.}
\]
The probability distribution monad $\mathcal{D}$ is a submonad of $\mathcal{M}_{[0,\infty]}$ \cite{brengos2015:jlamp}.

\subparagraph{Weighted languages and automata}
Fix an alphabet $\Sigma$ and a semiring $\mathcal{S}$. A weighted language (of finte words) is an $\mathcal{S}$-multiset with universe $\Sigma^\star$ and a weighted automaton $\mathcal{A}$ is (up to minor notational changes to the classical presentation \cite{DBLP:journals/iandc/Schutzenberger61b,DBLP:journals/iandc/BonchiBBRS12}) is a tuple $(m,\alpha,\chi)$ where $n$ is the set of states, $\alpha\colon n \to \mathcal{M_S}(\Sigma \times n)$ is the transition function (\ie a ternary $\mathcal{S}$-relation on $n\times\Sigma\times n$) and $(n,\chi\colon n \to \mathcal{S}$) is the multiset of final states. 
Given a weighted automaton $\mathcal{A} = (n,\alpha,\chi)$, the language accepted by a state $i$ of $\mathcal{A}$ is the multiset $L(\mathcal{A},i)$ with universe $\Sigma^\ast$ and membership function:
\[
	L(\mathcal{A},i)(\sigma_1\dots\sigma_k) = \sum \left\{\chi(j_k)\cdot \prod_{p< k} \hat\alpha(j_{p})(\sigma_{p+1},j_{p+1}) \,\middle|\, j_0 = i, j_1,\dots,j_k \in n \right\}
	\text{.}
\]

\subparagraph{Weighted automata and duality}

The multiset monad $\mathcal{M_S}$ and its Kleisli category do not meet conditions $(A)-(C)$ because of the cardinality constraint on multisets (which is ultimately due to $\mathcal{S}$ not admitting sums of arbitrary cardinality). There are two workarounds to this problem: one is to extend the definition of $\mathcal{M_S}$ to cover functions of arbitrary supports, the other is to realize that when dealing with regular behaviours as in (\ref{exp:regular}) we actually focus on systems over a finite state space. 
Hence, we restrict w.l.o.g.~to the subcategory identified by finite sets.
In this settings, the functor $(-)_-$ is readily defined by ``swapping columns and rows'' (maps in $\kl(\mathcal{M_S})$ can be regarded as $\mathcal{S}$):
\[
  f_-(y)(x) \defeq f(x)(y)\text{.}
\]
If $f$ is in the image of $\Set$ then the grade of each $f(x)$ is $1$. It follows that $f\circ f_- \leq \mathsf{id}$ and $f_-\circ f \geq \mathsf{id}$.

The free monad $\Sigma^\ast$ lifts to a monad $\overline{\Sigma^\ast}$ on $\kl(\mathcal{M_S})$ by the unique extension of the distributive law of the functor $\Sigma\times \mathcal{I}d$ over the monad $\mathcal{M_S}$ given by the assignment $\Sigma\times (X,\phi) \mapsto (\Sigma \times X, \sum\phi(x)\bullet(\sigma,x))$ \cite{brengos2015:jlamp}.
In this case, if $\Sigma^\ast$ is countable then any Eilenberg-Moore algebra $a\colon \overline{\Sigma^\ast} n \rightdcirc n=\Sigma^\ast\times n \to \mathcal{M_S}n$ on $\kl(\mathcal{M_S})$ yields a saturated system $a_-\colon  n\rightdcirc \overline{\Sigma^\ast} n = n\to \mathcal{M_S}({\Sigma^\ast \times n})$ by simple currying and uncurrying. Since $\overline{\Sigma^\ast}$ is the free monad over $\overline{\Sigma}$ on $\kl(\mathcal{M_S})$ \cite{brengos2015:lmcs}, the Eilenberg-Moore algebra $a$ is uniquely determined by $\underline{a}\colon \Sigma\times n \to \mathcal{M_S}n$ (\emph{conf} Subsec. \ref{subsection:the_duality}). Hence, so is its dual $a_-$.

\subsection{Fuzzy automata and their languages}
\label{section:appendix_examples_fuzzy}

\subparagraph{Quantales and fuzzy sets}

Let $(\mathcal{Q},\cdot ,1,\leq)$ be a \emph{unital quantale}, \ie. a relational structure with the property that: 
\begin{itemize}
\item $(\mathcal{Q},\cdot,1)$ is a monoid,
\item $(\mathcal{Q},\leq)$ is a complete lattice,
\item arbitrary suprema are preserved by the monoid multiplication.
\end{itemize}
In other words, a unital quantale is a monoid in the category $\Sup$ of join-preserving homomorphisms between complete join semi-lattices.
In the sequel we will often write $\perp_\mathcal{Q}$ or simply $\perp$ for the supremum of the empty set and denote a quantale $(\mathcal{Q},\cdot ,1,\leq)$ by its carrier set $\mathcal{Q}$, provided the associated structure is clear from the context.
Examples are given by booleans $(\{\bot,\top\},\land,\top,\implies)$, the real unit interval $([0,1],\cdot,1,\leq)$. More generally, for a monoid $(M,\odot,e)$ equipped with an order $\leq$ that is preserved by $\cdot$, the set $\mathcal{P}_{\downarrow}M$ of downward closed subsets of $M$ carries a unital quantale structure $(\mathcal{P}_{\downarrow}M,\cdot,1,\subseteq)$ where $1$ is the downward cone with cusp $e$:
\[ 
	1 = \{m \in M \mid m \leq e \}
\]
and $\cdot$ is the downward closure of the pointwise extension of $\odot$ to subsets of $M$:
\[
	X \cdot Y = \{m \mid \exists x \in X, \exists y \in Y (m \leq x \odot y)\}
	\text{.}
\]
A $\mathcal{Q}$-fuzzy set is a pair $(X,\phi)$ where $X$ is a set (called universe of discourse and often left implicit) and $\phi\colon X \to \mathcal{Q}$ is a a membership function. The value $\phi(x)$ is called \emph{grade of membership} of $x$. Taking the boolean quantale as $\mathcal{Q}$ yields ordinary sets and taking the real unit interval interval yields fuzzy sets in the classical sense.

The \emph{$\mathcal{Q}$-fuzzy powerset} functor $\mathcal{P_Q}\colon \Set \to \Set$ assigns to every set $X$ the set $\mathcal{P_Q}X = (X \to \mathcal{Q})$ of all fuzzy sets with universe $X$, and to every function $f\colon X \to Y$ the function mapping each $(X,\phi)$ to 
$(Y,\lambda y.\bigvee_{x:f(x)=y} \phi(x))$.
This functor carries a monad structure $(\mathcal{P_Q},\bigcup_\mathcal{Q},\{-\}_\mathcal{Q})$ when equipped with flattening and embedding into singletons. In particualr, the components of monadic multiplications and unit at $X$ are given by the assignments:
\[
  (\mathcal{P_Q}X,\psi) \mapsto (X,\lambda x.\vee_{\phi \in \mathcal{P_Q}X} \phi(x) \cdot \psi(\phi))
  \qquad
  (x \mapsto (X,\lambda y.\text{if } x=y \text{ then } 1 \text{ else } \bot)
  \text{.}
\]
The powerset monad $\mathcal{P}$ is a special case of the above where $\mathcal{Q}$ is the boolean quantale. 

\subparagraph{Fuzzy languages and automata}
Fix an alphabet $\Sigma$ and a quantale $\mathcal{Q}$.
A fuzzy language of (finite words) is a $\mathcal{Q}$-fuzzy set with universe $\Sigma^\ast$ and a fuzzy automaton (on finite words over $\Sigma$) is (up to minor notational changes to the classical presentation \cite{mordeson:2002fuzzy,doostfatemeh:2005fuzzy}) a tuple $\mathcal{A} = (n,\alpha,\chi)$ where $n$ is the state space, $\alpha\colon n \to \mathcal{P_Q}(\Sigma \times n)$ is the transition function (\ie a fuzzy ternary relation on $n\times\Sigma,\times n$), and $(n,\chi\colon n \to \mathcal{Q})$ is the fuzzy set of final states.
Given a fuzzy automaton $\mathcal{A} = (n,\alpha,\chi)$, the language accepted by a state $i$ of $\mathcal{A}$ is the fuzzy set $L(\mathcal{A},i)$ with universe $\Sigma^\ast$ and membership function:
\[
L(\mathcal{A},i)(\sigma_1\dots\sigma_k) = \bigvee \left\{\chi(j_k)\cdot \prod_{p< k} \alpha(j_{p})(\sigma_{p+1},j_{p+1}) \,\middle|\, j_0 = i, j_1,\dots,j_k \in n \right\}\text.
\]

\subparagraph{Fuzzy automata and duality}
The fuzzy powerset monad and its Kleisli category meet conditions $(A)-(C)$ from \cref{section:duality} when $\Set$ is taken as $\cat{J}$. This follows by the simple observation that the Kleisli category of $\mathcal{P_Q}$ is isomorphic to the category $\cat{Mat}-\mathcal{Q}$ of $\mathcal{Q}$-valued matrices. The functor $(-)_-$ is readily defined by ``swapping columns and rows'':
\[
  f_-(y)(x) \defeq f(x)(y)\text{.}
\]
If $f$ is in the image of $\Set$ then the grade of each $f(x)$ is $1$. It follows that $f\circ f_- \leq \mathsf{id}$ and $f_-\circ f \geq \mathsf{id}$.

The free monad $\Sigma^\ast$ lifts to a monad $\overline{\Sigma^\ast}$ on $\kl(\mathcal{P_Q})$ by the unique extension of the distributive law of the functor $\Sigma\times \mathcal{I}d$ over the monad $\mathcal{P_Q}$ given by the assignment $(\Sigma \times (X,\phi)) \mapsto (\Sigma \times X, \lambda (\sigma,x).\phi(x))$. The monad $\overline{\Sigma^\ast}$ meets condition $(D)$ from \cref{section:duality} \cite{brengos2019:lmcs}. In particular, saturation is defined on every $\alpha$ as $\alpha^\ast(x)(y) = \bigvee_{k < \omega} \alpha^k$ (where $\alpha^0 = \mathsf{id}$ and $\alpha^{k+1} = \alpha \circ \alpha^k$).
It follows from \cref{proposition:duality_free_monads} that every Eilenberg-Moore algebra for $\overline{\Sigma^\ast}$ is dual to the saturations of morphism of the form  $X\to \mathcal{P_Q}(\Sigma \times X)\hookrightarrow \mathcal{P_Q}(\Sigma^\ast \times X)$ and it follows by construction of $\overline{\Sigma^\ast}$ that any $\alpha \in \mathcal{SAT}(\overline{\Sigma^\ast})$ h map is determined is equivalently defined as a map $\hat{\alpha}\colon n\to \mathcal{P_Q}(\Sigma \times n)$ \ie the transition map of a fuzzy automaton. As a consequence, the language $L(\mathcal{A},i)$ accepted by a state $i$ of a fuzzy automaton $\mathcal{A} = (n,\alpha,\chi)$ is given by the regular map:
\[
 L(\mathcal{A},i) = 1\rightdcirc n \stackrel{(\nu_X \circ \hat{\alpha})^\ast}{\rightdcirc } \overline{\Sigma^\ast}n \stackrel{\overline{\Sigma^\ast}\chi}{\rightdcirc }\overline{\Sigma^\ast}1
\]
where $\nu$ is the canonical embedding of $\mathcal{P_Q}(\Sigma\times \mathcal{I}d)$ into $\mathcal{P_Q}(\Sigma^\ast\times \mathcal{I}d)$.

\end{document}